\documentclass[12pt]{article}

\usepackage{amsmath}
\usepackage{amssymb}
\usepackage{latexsym}
\usepackage[colorlinks=true,citecolor=blue,urlcolor=black,linkcolor=blue]{hyperref}
\usepackage{url}

\usepackage[all]{xy}

\newcommand{\drpullback}[1][dr]{\save*!/#1-1.2pc/#1:(-1,1)@^{|-}\restore}

\newcommand{\name}[1]{\ulcorner #1 \urcorner}

\newcommand{\Bij}{\kat{Bij}}
\newcommand{\Grpd}{\kat{Grpd}}
\newcommand{\TEmb}{\kat{TEmb}}
\newcommand{\id}{\operatorname{id}}

\newcommand{\Map}{\operatorname{Map}}
\providecommand{\norm}[1]{\left| {#1}\right|}
\newcommand{\comma}{\raisebox{1pt}{$\downarrow$}}

\newcommand{\tree}{\mathbf{T}}
\newcommand{\forest}{\mathbf{F}}

\def\xxgraph{\mathbf G}
\def\nest{\mathbf N}

\newcommand{\brok}[2]{{\textstyle{\frac{#1}{#2}}}}

\usepackage[noBBpl]{mathpazo}

\renewcommand{\texttt}[1]{{\fontfamily{pcr}\fontseries{m}\fontshape{n}%
\selectfont #1}}

\usepackage{mathrsfs}
\providecommand{\lastUpdate}[1]{#1}

\newcommand{\HH}{\mathscr{H}}
\newcommand{\CK}{\mathscr{H}_{\mathrm{CK}}{}}
\newcommand{\BB}{\mathscr{B}}
\newcommand{\CC}{\mathscr{C}}

\newcommand{\B}{\mathbb{B}}
\newcommand{\N}{\mathbb{N}}
\newcommand{\C}{\mathbb{C}}
\newcommand{\Q}{\mathbb{Q}}
\newcommand{\isopil}{\stackrel{\raisebox{0.1ex}[0ex][0ex]{\(\sim\)}}%
			{\raisebox{-0.15ex}[0.28ex]{\(\rightarrow\)}}}

\newcommand{\isopilback}{\stackrel{\raisebox{0.1ex}[0ex][0ex]{\(\sim\)}}%
			{\raisebox{-0.15ex}[0.28ex]{\(\leftarrow\)}}}

\newcommand{\upperstar}{^{\raisebox{-0.25ex}[0ex][0ex]{\(\ast\)}}}
\newcommand{\lowerstar}{_{\raisebox{-0.33ex}[-0.5ex][0ex]{\(\ast\)}}}
\newcommand{\lowershriek}{_!}

\newcommand{\tensor}	{\otimes}

\newcommand{\Id}{\operatorname{Id}}
\newcommand{\Aut}{\operatorname{Aut}}
\newcommand{\res}{\operatorname{res}}

\newcommand{\DD}{\mathscr{D}}

\providecommand{\kat}[1]{\text{\textbf{\textsl{#1}}}}

\newcommand{\Set}{\kat{Set}}
\newcommand{\alg}{\textbf{-}\mathbf{alg}}

\newcommand{\Z}{\mathbb{Z}}

\newcommand{\qed}{\hspace*{\fill}\nolinebreak$\Box$}

\makeatletter

\renewcommand*\l@section[2]{%
  \ifnum \c@tocdepth >\z@
    \addpenalty\@secpenalty
    \addvspace{0.4em \@plus\p@}
    \setlength\@tempdima{1.5em}%
    \begingroup
      \parindent \z@ \rightskip \@pnumwidth
      \parfillskip -\@pnumwidth
      \leavevmode \bfseries
      \advance\leftskip\@tempdima
      \hskip -\leftskip
      #1\nobreak\hfil \nobreak\hb@xt@\@pnumwidth{\hss #2}\par
    \endgroup
  \fi}

\renewcommand{\tableofcontents}{%
   \begin{center}
\begin{minipage}{12cm}
   \begin{center}
     
     \vspace{4pt}
     
       \bf{\contentsname}
   \end{center}
	
       \vspace{-18pt}
	
   \footnotesize
   \begin{center}
\@starttoc{toc}
   \end{center}	
\end{minipage}
	\end{center}
	\addvspace{2em \@plus\p@}
}

\makeatother

\usepackage{theorem}
\theorembodyfont{\normalfont\itshape}
\theoremstyle{change}
\newtheorem{lemma}{Lemma.}[section]
\newtheorem{prop}[lemma]{Proposition.}
\newtheorem{theorem}[lemma]{Theorem.}
\newtheorem{cor}[lemma]{Corollary.}
\theorembodyfont{\normalfont}

\newtheorem{taller}[lemma]{$\!\!$}
\newenvironment{blanko}[1]%
{\begin{taller}{\normalfont\bfseries #1}\normalfont}%
{\end{taller}}

\providecommand{\qed}{\hspace*{\fill}$\Box$}

\newenvironment{proof}{%
\begin{list}{\em Proof. }%
{\setlength{\labelsep}{0mm}\setlength{\leftmargin}{0mm}%
\setlength{\labelwidth}{0mm}\setlength{\listparindent}{\parindent}%
\setlength{\parsep}{\parskip}\setlength{\partopsep}{0mm}}%
\item}{\qed\end{list}}

\newenvironment{proof*}[1]{%
\begin{list}{\em #1 }%
{\setlength{\labelsep}{0mm}\setlength{\leftmargin}{0mm}%
\setlength{\labelwidth}{0mm}\setlength{\listparindent}{\parindent}%
\setlength{\parsep}{\parskip}\setlength{\partopsep}{0mm}}%
\item}{\qed\end{list}}

\newcommand{\overskrift}[1]{\par\noindent\relax{\LARGE #1}\par\bigskip}

\usepackage{texdraw}
\input{txdtools}

\everytexdraw{%
	\drawdim pt \linewd 0.5 \textref h:C v:C
	\setunitscale 1
}

\newcommand{\onedot}{
  \bsegment
    \move (0 0) \fcir f:0 r:2
  \esegment
}
\newcommand{\Onedot}{
  \bsegment
    \move (0 0) \fcir f:0 r:2.5
  \esegment
}

  \newcommand{\inlineDotlessTree}{%
\raisebox{-4pt}{
\begin{texdraw} \linewd 0.5 \bsegment
    \move (0 0) \lvec (0 15) \move (1 0)
  \esegment \end{texdraw} } }

\newcommand{\trekant}{
\bsegment \setunitscale{0.8}
\move (1 0) \lvec (6 0) \smalldot \lvec (16 10)
\move (6 0) \lvec (16 -10)
\move (12 6) \smalldot  \lvec (12 -6) \smalldot
\esegment
}

\newcommand{\trevert}{
\bsegment
\move (3 0) \lvec (6 0) \lvec (9 3) \move (6 0) \lvec (9 -3)
\esegment
}

\newcommand{\tovert}{
\bsegment 
\move (1 0) \lvec (6 0) \smalldot \lvec (11 0)
\esegment
}

\newcommand{\hhgraph}{
  \bsegment
	\linewd 1
      \move (-30 0) \lvec (-15 0) \Onedot
      \move (30 0) \lvec (15 0) \Onedot      
      \move (0 0) \lcir r:15
      \move (-5 14) \Onedot
      \move (5 14) \Onedot
      \move (0 14) \larc r:5 sd:180 ed:360
      \move (-5 -14) \Onedot
      \move (5 -14) \Onedot
      \move (0 -14) \larc r:5 sd:00 ed:180
  \esegment
}
\newcommand{\redoval}{
  \bsegment
  	\linewd 1
        \red
      \lpatt (1 2)
      \move (0 0) \freeEllipsis{12}{7}{0}
\esegment
}
\newcommand{\bigredoval}{
  \bsegment
  	\linewd 1
        \red
      \lpatt (1 2)
      \move (0 0) \freeEllipsis{23}{28}{0}
\esegment
}

\newcommand{\vgraph}{
  \bsegment
	\linewd 1
      \move (-30 0) \lvec (-15 0) \Onedot
      \move (30 0) \lvec (15 0) \Onedot      
      \move (0 0) \lcir r:15
      \move (0 15) \Onedot \lvec (0 -15) \Onedot
  \esegment
}
\newcommand{\redvoval}{
  \bsegment
  	\linewd 1
        \red
      \lpatt (1 2)
      \move (0 0) \freeEllipsis{14}{23}{0}
\esegment
}

\newcommand{\edge}{
\bsegment 
  \rlvec (0 10) \onedot 
  \savepos (0 10)(*ex *ey)
\esegment
\move (*ex *ey)
}

\newcommand{\edgesV}{
\bsegment
  \move (0 0)
  \rlvec (-4 10) \onedot
  \move (0 0)
  \rlvec (4 10) \onedot
  \savepos (4 10)(*ex *ey)
\esegment
\move (*ex *ey)
}

\newcommand{\edgesW}{
\bsegment
  \move (0 0)
  \rlvec (-6 10) \onedot
  \move (0 0)
  \rlvec (0 10) \onedot
  \move (0 0)
  \rlvec (6 10) \onedot
  \savepos (6 10)(*ex *ey)
\esegment
\move (*ex *ey)
}

\newcommand{\ctreeone}{
\begin{texdraw}
  \move (0 0) \onedot
\end{texdraw}
}

\newcommand{\ctreetwo}{
\begin{texdraw}
  \move (0 0) \onedot
  \edge
\end{texdraw}
}

\newcommand{\ctreethreeV}{
\begin{texdraw}
  \move (0 0) \onedot
  \edgesV
\end{texdraw}
}

\newcommand{\ctreethreeL}{
\begin{texdraw}
  \move (0 0) \onedot
  \edge \edge
\end{texdraw}
}

\newcommand{\ctreefourL}{
\begin{texdraw}
  \move (0 0) \onedot
  \edge \edge \edge
\end{texdraw}
}

\newcommand{\ctreefourY}{
\begin{texdraw}
  \move (0 0) \onedot
  \edge \edgesV
\end{texdraw}
}

\newcommand{\ctreefourVL}{
\begin{texdraw}
  \move (0 0) \onedot
  \edgesV \move (-4 10) \edge
\end{texdraw}
}
\newcommand{\ctreefourVR}{
\begin{texdraw}
  \move (0 0) \onedot
  \edgesV \edge
\end{texdraw}
}

\newcommand{\ctreefourW}{
\begin{texdraw}
  \move (0 0) \onedot
  \edgesW
\end{texdraw}
}

\newcommand{\ctreefiveL}{
\begin{texdraw}
  \move (0 0) \onedot
  \edge \edge \edge \edge
\end{texdraw}
}

\newcommand{\ctreefiveIY}{
\begin{texdraw}
  \move (0 0) \onedot
  \edge \edge \edgesV
\end{texdraw}
}

\newcommand{\ctreefiveVII}{
\begin{texdraw}
  \move (0 0) \onedot
  \edgesV \edge \move (-4 10) \edge
\end{texdraw}
}
\newcommand{\ctreefiveYL}{
\begin{texdraw}
  \move (0 0) \onedot
  \edge \edgesV \move (-4 20) \edge
\end{texdraw}
}
\newcommand{\ctreefiveYR}{
\begin{texdraw}
  \move (0 0) \onedot
  \edge \edgesV \edge
\end{texdraw}
}

\newcommand{\ctreefiveVLL}{
\begin{texdraw}
  \move (0 0) \onedot
  \edgesV \move (-4 10) \edge \edge
\end{texdraw}
}
\newcommand{\ctreefiveVRR}{
\begin{texdraw}
  \move (0 0) \onedot
  \edgesV \edge \edge
\end{texdraw}
}

\newcommand{\ctreefiveVlV}{
\begin{texdraw}
  \move (0 0) \onedot
  \edgesV \move (-4 10) \edgesV
\end{texdraw}
}
\newcommand{\ctreefiveVrV}{
\begin{texdraw}
  \move (0 0) \onedot
  \edgesV \edgesV
\end{texdraw}
}

\newcommand{\ctreefiveIW}{
\begin{texdraw}
  \move (0 0) \onedot
  \edge \edgesW 
\end{texdraw}
}

\newcommand{\ctreefiveWL}{
\begin{texdraw}
  \move (0 0) \onedot
  \edgesW 
  \move (-6 10) \edge
\end{texdraw}
}

\newcommand{\ctreefivewide}{
\begin{texdraw}
  \move (0 0) 
  \bsegment
  \move (0 0) \onedot
  \rlvec (-10 10) \onedot
  \move (0 0)
  \rlvec (-3.2 10) \onedot
  \move (0 0)
  \rlvec (3.2 10) \onedot
  \move (0 0)
  \rlvec (10 10) \onedot
  \esegment
\end{texdraw}
}

\newcommand{\smalldot}{
  \bsegment
    \move (0 0) \fcir f:0 r:1.5
  \esegment
}

\newcommand{\tinydot}{
  \bsegment
    \move (0 0) \fcir f:0 r:1.1
  \esegment
}

\newcommand{\twotree}{
\bsegment 
\move (0 -4) \lvec (0 0) \tinydot
\lvec (-2 5) \tinydot
\lvec (-4 10)
\move (-2 5) \lvec (0 10)
\move (0 0) \lvec (3 10)
\esegment
}

\newcommand{\onetree}{
\bsegment 
\move (0 -4) \lvec (0 0) \tinydot
\lvec (-2 5) 
\move (0 0) \lvec (2 5)
\esegment
}

\newcommand{\zerotree}{
\bsegment 
\move (0 -4) \lvec (0 4)
\esegment
}

\newcommand{\twonodetree}{
\bsegment 
\move (0 -4) \smalldot \lvec (0 4) \smalldot
\esegment
}

\newcommand{\vtree}{
\bsegment 
\move (0 -4) \smalldot \lvec (-3 4) \smalldot
\move (0 -4) \lvec (3 4) \smalldot
\esegment
}

\makeatletter
\renewcommand{\ps@headings}
	{\setlength{\headheight}{13pt}%
	 \setlength{\headsep}{12pt}%
	 \renewcommand{\@oddhead}{\parbox{\textwidth}{%
			\footnotesize
			\texttt{\jobname.tex \ \ \lastUpdate{2017-02-01 21:21} \ \ \ 
			\hfill [\thepage/\pageref{lastpage}]}
			\\ \rule[8pt]{\textwidth}{0.3pt}}%
	 }
	\renewcommand{\@oddfoot}{}
	\renewcommand{\@evenfoot}{}%
}
\makeatother

\newcommand{\red}{\writeps{1 0 0 setrgbcolor}}

\newcommand{\tokant}{
\bsegment \setunitscale{0.8}

\move (0 -10) \lvec (0 10)
      \move (0 5) \smalldot
      \move (0 -5) \smalldot
      \move (0 0) \larc r:5 sd:-90 ed:90

\esegment
}

\newcommand{\rundtokant}{
\bsegment \setunitscale{0.8}
\move (2 0) \lvec (6 0) \smalldot \move (15 0) \smalldot \lvec (19 0)
\move (10.5 0) \lcir r:4.25
\esegment
}

\newcommand{\freeEllipsis}[3]{
    \writeps {
      gsave 
      #3 rotate 
    }
    \lellip rx:#1 ry:#2
    \writeps { 
      grestore
    }
}

\begin{document}

\pagestyle{headings}

\vspace*{15pt}

\begin{center}
  
\overskrift{Polynomial functors and combinatorial Dyson--Schwinger equations}

\bigskip

\textsc{Joachim Kock}\footnote{
Departament de Matem\`atiques,
Universitat Aut\`onoma de
Barcelona,
08193 Bellaterra,
Spain}
\end{center}

\begin{abstract}\footnotesize
  We present a general abstract framework for combinatorial Dyson--Schwinger
  equations, in which combinatorial identities are lifted to explicit bijections
  of sets, and more generally equivalences of groupoids.  Key features of
  combinatorial Dyson--Schwinger equations are revealed to follow from general
  categorical constructions and universal properties.  Rather than beginning with
  an equation inside a given Hopf algebra and referring to given Hochschild
  $1$-cocycles, our starting point is an abstract fixpoint equation in
  groupoids, shown canonically to generate all the algebraic structure.
  Precisely, for any finitary polynomial endofunctor $P$ defined over groupoids,
  the system of combinatorial Dyson--Schwinger equations $X=1+P(X)$ has a
  universal solution, namely the groupoid of $P$-trees.  The isoclasses of
  $P$-trees generate naturally a Connes--Kreimer-like bialgebra, in which the
  abstract Dyson--Schwinger equation can be internalised in terms of canonical
  $B_+$-operators.  The solution to this equation is a series (the Green
  function) which always enjoys a Fa\`a di Bruno formula, and hence generates a
  sub-bialgebra isomorphic to the Fa\`a di Bruno bialgebra.  Varying $P$ yields
  different bialgebras, and cartesian natural transformations between various
  $P$ yield bialgebra homomorphisms and sub-bialgebras, corresponding for
  example to truncation of Dyson--Schwinger equations.  Finally, all
  constructions can be pushed inside the classical Connes--Kreimer Hopf algebra
  of trees by the operation of taking core of $P$-trees.  A byproduct of the
  theory is an interpretation of combinatorial Green functions as inductive
  data types in the sense of Martin-L\"of Type Theory (expounded elsewhere).
\end{abstract}

%



\footnotesize

\tableofcontents

\normalsize

\setcounter{section}{-1}
\section{Introduction}

The Dyson--Schwinger equations are infinite hierarchies of integral equations
relating different Green functions of a given quantum field theory.  They are
often referred to as quantum equations of motion, as they can be expressed, in
the spirit of Schwinger's 1951 paper~\cite{Schwinger:1951}, in terms of
functional integrals and a least action principle --- a quantum version of the
Euler--Lagrange equations.  A standard reference for Schwinger's approach
is Itzykson--Zuber~\cite{Itzykson-Zuber:QFT}.

A more intuitive approach is in terms of skeleton expansions, in the spirit of
Dyson's 1949 paper~\cite{Dyson:S-matrix}, highlighting the recursive nature of
Feynman diagrams.  In this approach, for which we refer to
Bjorken--Drell~\cite{Bjorken-Drell:fields}, 
the equations for the 1PI Green functions are often
rendered graphically, such as
  \begin{center}\begin{texdraw}
    \move (0 0) 
    \bsegment
    \lvec (12 0) \rlvec (8 8) \move (12 0) \rlvec (8 -8)
    \move (12 0) \fcir f:0.8  r:3  \lcir r:3
    \esegment
    
    \htext (32 0) {$=$}
    
       \move (40 0) \bsegment
       \move (2 0)
       \lvec (12 0) \rlvec (8 8) \move (12 0) \rlvec (8 -8) 
       \esegment
       
        \htext (70 0) {$+$}
   
	\move (82 0)
	\bsegment
       \move (0 0)
       \lvec (12 0) \rlvec (18 18) \move (12 0) \rlvec (18 -18) 
       \move (22 10) \lvec (22 -10)
       \move (12 0) \fcir f:0.8  r:3  \lcir r:3
       \move (22 10) \fcir f:0.8  r:3  \lcir r:3
       \move (22 -10) \fcir f:0.8  r:3  \lcir r:3
       \esegment
       
        \htext (126 0) {$+ \ \ \frac{1}{2!}$}
   
	\move (140 0)
	\bsegment
       \move (0 0)
       \lvec (12 0) \rlvec (28 28) \move (12 0) \rlvec (28 -28) 
       \move (22 10) \lvec (32 -20)
       \move (27 0) \fcir f:1 r:3
       \move (22 -10) \lvec (32 20)
       \move (12 0) \fcir f:0.8  r:3  \lcir r:3
       \move (22 10) \fcir f:0.8  r:3  \lcir r:3
       \move (22 -10) \fcir f:0.8  r:3  \lcir r:3
       \move (32 20) \fcir f:0.8  r:3  \lcir r:3
       \move (32 -20) \fcir f:0.8  r:3  \lcir r:3
       \esegment
       
            \htext (208 0) {$+ \quad \cdots$}
  \end{texdraw}
  \end{center}
with one term for each primitive graph (skeleton).
Each term is regarded as a shorthand for an infinite sum of Feynman integrals
with a common kernel.

For a long period of time, the combinatorial aspect of QFT was characteristic
for {\em perturbative} QFT, and is intimately related with Renormalisation
Theory.  From Feynman via Bogoliubov, Parasiuk, Hepp, and Zimmermann, a
culmination was achieved in the work of Kreimer~\cite{Kreimer:9707029} and his
collaborators around the turn of the millennium, when this combinatorics was
distilled into clear-cut algebraic structures with numerous connections to many
fields of Mathematics~\cite{Connes-Kreimer:9808042},
\cite{Connes-Kreimer:9912092}, \cite{EbrahimiFard-Guo-Manchon:0602004}.
Specifically, Kreimer~\cite{Kreimer:9707029} discovered that the combinatorics
of perturbative renormalisation is encoded in a Hopf algebra of trees, now
called the Connes--Kreimer Hopf algebra.

However, it soon became clear (see for example \cite{Kreimer:0306020},
\cite{Bergbauer-Kreimer:0506190}, \cite{Kreimer:0509135},
\cite{Kreimer-Yeats:0605096}, \cite{vanSuijlekom:0807}) that this combinatorial
and algebraic insight is also valuable in the {\em non-perturbative} regime,
when combined with the renormalisation group.  Abstracting away the Feynman
rules from the Dyson--Schwinger equations, Kreimer~\cite{Kreimer:0306020} and
Bergbauer--Kreimer~\cite{Bergbauer-Kreimer:0506190} initiated the study of the
{\em combinatorial Dyson--Schwinger equations}, which are the main object of the
present contribution. 
To formulate the combinatorial Dyson--Schwinger equations,
Bergbauer--Kreimer required an ambient combinatorial Hopf algebra (typically the
Connes--Kreimer Hopf algebra), equipped with certain Hochschild
$1$-cocycles, as reviewed below.  While these Hopf algebras of graphs or trees
belong to the perturbative regime, they were shown to contain smaller Hopf
algebras spanned by the solutions to the combinatorial Dyson--Schwinger
equations (cf.~Theorem~\ref{thm:BK} below), which can be given also a 
non-perturbative meaning (see~\cite{Kreimer:0609004}).

\bigskip

The present contribution takes a further abstraction step, showing that also the
algebra can be dispensed with, so that only pure combinatorics is left,
accentuating the patent recursive aspect of the equations, which is here related
to fundamental constructions in Category Theory.  One may say that this is the
natural setting for the combinatorial Dyson--Schwinger equations, in the sense
that their essential intuitive content is modelled directly, without reference
to Hopf algebras or Hochschild cocycles.  It is shown that in fact the rich
algebraic structures can be {\em derived} from the abstract equation.

In this work, instead of considering combinatorial Dyson--Schwinger
equations inside preexisting combinatorial Hopf algebras, we start with an
abstract polynomial fixpoint equation formulated in the category of groupoids.
Each such equation, specified in terms of a finitary polynomial endofunctor $P$,
has a universal solution, namely the groupoid of $P$-trees (the homotopy initial
$(1+P)$-algebra).  This solution canonically {\em defines} a combinatorial
bialgebra $\BB_P$, by a slight variation of the Connes--Kreimer construction
\cite{Kock:1109.5785}, which becomes a canonical, almost tautological, home for
the further algebraic structures: again by general principles this bialgebra
contains a series, the Green function of all the generators of the bialgebra,
weighted by their symmetry factors, which is the solution to a version of the
Dyson--Schwinger equation internalised to $\BB_P$, formulated in terms of
canonical $B_+$-operators easily extracted from the groupoid equivalence
expressing the solution.
Furthermore, by a general result of \cite{GalvezCarrillo-Kock-Tonks:1207.6404}, 
the homogeneous pieces of this
series always generates a sub-bialgebra isomorphic to the Fa\`a di Bruno
bialgebra.  

Different choices of $P$ yield different bialgebras, and they are all interconnected
by bialgebra homomorphisms induced by cartesian natural transformations between the
polynomial endofunctors, and each containing a canonical copy of
the Fa\`a di Bruno bialgebra.  A special case of this is the case of cartesian
subfunctors, which correspond to {\em truncation} of Dyson--Schwinger equations,
a topic of high current interest, especially in Quantum Chromodynamics~\cite{Roberts:1203.5341}.
Finally, every such bialgebra of $P$-trees comes with a
canonical bialgebra homomorphism to the Connes--Kreimer Hopf algebra, which
therefore also receives a plethora of different Fa\`a di Bruno sub-bialgebras.

\bigskip

It is important to note that the trees arising naturally from the abstract 
equations are {\em
operadic} trees (i.e.~trees with open-ended leaves and root) rather than the
combinatorial trees commonly used in the literature.  
With combinatorial trees,
the only grading is by the number of nodes, but the natural grading for the
Dyson--Schwinger equations and the Fa\`a di Bruno formula is actually the {\em
operadic grading}, which is by the number of leaves (minus the number of roots),
a grading that cannot be seen at the level of combinatorial trees. 
When pushed into the Connes--Kreimer Hopf algebra of trees, the Fa\`a di Bruno
bialgebras become Hopf algebras, but they are not generated by homogeneous
elements for the node grading, in contrast to the infinite families of Fa\`a di 
Bruno Hopf sub-algebras constructed by Foissy~\cite{Foissy:0707.1204}.

The presence of leaves is key to the simplicity of the derivation of the
results, an insight going back to \cite{Kock:1109.5785} and exploited further in
\cite{GalvezCarrillo-Kock-Tonks:1207.6404}.  The price to pay is that we have
to work in bialgebras, not Hopf algebras, and that the $B_+$-operators are not
Hochschild $1$-cocycles.  However, it has been observed recently~\cite{Kock:1411.3098}
that the BPHZ renormalisation procedure generalises from Hopf algebras to the
kind of bialgebras arising here, which have in particular the property that the
zeroth graded piece is spanned by group-like elements.  Similarly, the crucial
features which Kreimer and his collaborators derive from the Hochschild
$1$-cocycle condition --- locality of counterterms and finite renormalisation,
as well as the Fa\`a di Bruno formula --- can also be derived from the
not-quite-$1$-cocycle $B_+$-operators.

\bigskip

The theory of polynomial functors, the main vehicle for the present approach to
combinatorial Dyson--Schwinger equations, has roots in Topology, Representation
Theory, Combinatorics, Logic, and Computer Science.  A standard reference is
\cite{Gambino-Kock:0906.4931}, which also contains pointers to those original
developments.  In many cases the polynomial functors serve to parametrise
operations of some sort, and this is also the case here, where they organise
$B_+$-operators, hence formalising operadic interpretations hinted at in many
papers by Kreimer (e.g.~\cite{Kreimer:0010059}, \cite{Kreimer:0202110}, 
\cite{Kreimer:0609004}), perhaps most explicitly with
Bergbauer~\cite{Bergbauer-Kreimer:0506190}.  It should be mentioned that
operads can be absorbed into the theory of polynomial 
functors~\cite{Gambino-Kock:0906.4931}, \cite{Weber:1412.7599}.

One pleasant aspect of the polynomial formalism is the ease of the passage from
the one-variable case to the many-variable case, such as found in realistic
quantum field theory: the formulae, and to a large extent also the proofs, look
the same in the many-variable case, only do the symbols refer to more
complicated structures.

The notion of $P$-tree, already studied in category theory \cite{Kock:0807},
\cite{Kock:MFPS28}, \cite{Kock:1109.5785},
\cite{Kock-Joyal-Batanin-Mascari:0706}, has some conceptual advantages over the
notions of decorated trees usual employed in the QFT literature.  In particular,
they feature meaningful symmetry factors with regard to the Feynman graphs they
are a recipe for (\cite{Kock:graphs-and-trees}), and results for graphs can be
derived from the $P$-tree formalism.  For this, it is essential to work over
groupoids, and not just over sets.

Hand in hand with the formalism of polynomial functors and groupoids goes the
possibility of lifting the results from their classical form of algebraic
identities to the `objective' level where they amount to bijections of sets and
equivalences of groupoids.  The algebraic results follow from these by taking
homotopy cardinality, but the bijective proofs represent deeper insight,
unveiling structural patterns that cannot be seen on the algebraic level, such
as universal properties.  (Historically, a starting point for `objective'
combinatorics is Joyal's theory of species~\cite{Joyal:1981}.)  Specifically, for
the present results, the solution to the Dyson--Schwinger equation takes the
form of an initial object in the category of ($1+P)$-algebras; the resulting
bialgebra is
an instance of the general construction of incidence (co)algebras from simplicial
groupoids, which can be formulated at the objective level of slices of the 
category of groupoids~\cite{Galvez-Kock-Tonks:1512.07573}; and the Fa\`a di Bruno formula
essentially drops out as an instance of a homotopy version of the double
counting principle~\cite{GalvezCarrillo-Kock-Tonks:1207.6404}, itself an 
instance of commutation of colimits.

The existence of initial algebras, and hence solutions of polynomial fixpoint
equations over sets, is classical insight in Category Theory, going back at
least to Lambek's 1968 paper \cite{Lambek:fixpoint}.  Subsequent developments of
this theory were driven by applications to Logic and Computer Science: initial
algebras provide semantics for inductive data types.  More precisely, under the
so-called Seely correspondence between locally cartesian closed categories and
Martin-L\"of Type Theory, initial algebras for polynomial functors are precisely
the W-types (also called wellfounded trees)~\cite{Moerdijk-Palmgren:Wellfounded}
of Martin-L\"of Type Theory.  The upgrade from sets to groupoids is motivated in
part by practical considerations in Computer Science related to data types with
symmetries (see \cite{Kock:MFPS28}), and in part by developments in Homotopy
Type Theory~\cite{HoTT-book} where ultimately $\infty$-groupoids is the real stuff.

The close connections with Type Theory is an important byproduct of the present
contribution, expanded upon in a companion paper~\cite{Kock:DSE-W}.  In
short, Green functions are revealed to be inductive types, in a precise
technical sense, thus formalising the classical wisdom that Dyson--Schwinger
equations express self-similarity properties of their solutions.

\bigskip

Since polynomial functors and groupoids are not assumed to be part of
the standard toolbox of this paper's intended core readership of mathematical 
physicists, an effort is made to explain these notions along the way. 
Thus, the natural level of generality of the
main results is not arrived at until Sections~\ref{sec:groupoid-DSE} and
\ref{sec:groupoid-Green}; building up to these results there are four
sections with the necessary categorical background on polynomial functors. 

However, it is actually possible to explain the key ideas with only modest
categorical background, by restricting to the case of polynomial functors
over $\Set$ in one variable.  For the sake of getting quickly to the main ideas,
this case is treated first, aiming also at motivating the heavier
machinery treated in turn.  The slight duplication of arguments resulting
from this preview will hopefully prove worthwhile.

Altogether, these considerations have led to the following arrangement of
the material into sections.

In Section~\ref{sec:DSE} we briefly recall the combinatorial Dyson--Schwinger
equations in their prototypical form in the Connes--Kreimer Hopf algebra
of rooted trees.  In Section~\ref{sec:1var} the main results are outlined in
the simplest case, that of polynomial functors over $\Set$
in one variable. 
As preparation for the main results in Section~\ref{sec:groupoid-DSE} and
\ref{sec:groupoid-Green}, many-variable functors are introduced in
Section~\ref{sec:many}, the needed notions from the (homotopy) theory of
groupoids in Section~\ref{sec:groupoids}, and groupoid polynomial functors in
Section~\ref{sec:P}.  After the main results in Section~\ref{sec:groupoid-DSE}
and \ref{sec:groupoid-Green}, the short Section~\ref{sec:nat} deals with
functoriality, observing that cartesian morphisms of polynomial functors induce
bialgebra homomorphism compatible with $B_+$-operators and Green functions in
the expected ways.  In Section~\ref{sec:Foissy} we take a look at the different
form of Dyson--Schwinger equation studied by Foissy~\cite{Foissy:0707.1204}.
Section~\ref{sec:trees-graphs}, a preview of a paper in
preparation~\cite{Kock:graphs-and-trees}, explains how to encode Feynman graphs
as $P$-trees, and how to transfer results to the realm of graphs.

\bigskip

\noindent
{\bf Acknowledgments.}
  It was Kurusch Ebrahimi-Fard who first got me interested in Quantum Field 
  Theory and Renormalisation
  seven years ago, and his continuing guidance and help since then has been essential for
  this work, and is gratefully acknowledged.  This paper is an elaboration upon
  my talk at the workshop on {\em Mathematical Aspects of Hadron Physics} at
  the ECT* in Trento, October 2012, and I thank the organisers, Kurusch
  Ebrahimi-Fard, Fr\'ed\'eric Patras, and Craig Roberts for the invitation to
  speak, and for many instructive discussions on Dyson--Schwinger equations and
  related topics.  More recently I have benefited from conversations with 
  Lutz Klaczynski and Michael Borinsky in Berlin.
  Finally, I acknowledge support from grant number
  MTM2013-42293-P of Spain.

\section{Combinatorial Dyson--Schwinger equations}
\label{sec:DSE}

\begin{blanko}{The Connes--Kreimer Hopf algebra of trees 
  \cite{Connes-Kreimer:9808042}.}\label{CK}
  The {\em Connes--Kreimer Hopf algebra} of (rooted) trees (also called the
  {\em Butcher--Connes--Kreimer Hopf algebra})
  is the free commutative
  algebra $\CK$ on the set of isomorphism classes of
  {\em combinatorial trees} such as \
  \raisebox{1pt}{\begin{texdraw}\smalldot\end{texdraw}} , \raisebox{-3pt}{
    \begin{texdraw}\twonodetree\end{texdraw}} , \raisebox{-3pt}{
    \begin{texdraw}\vtree\end{texdraw}} .
    (`Combinatorial' as opposed to the operadic trees (\ref{operadictrees})
    that will play an important role in what follows.)
    The comultiplication is
  given on generators by
  \begin{eqnarray*}
  \Delta:  \CK & \longrightarrow & \CK \otimes \CK  \\
    T & \longmapsto & \sum_c P_c \otimes S_c ,
  \end{eqnarray*}
  where the sum is over all admissible cuts of $T$; the left-hand factor $P_c$
  is the forest (interpreted as a monomial) found above the cut, and $S_c$ is
  the subtree found below the cut (or the empty forest, in case the cut is below
  the root).   Admissible cut means: either a subtree containing the root, or the empty set.
  $\CK$ is a connected bialgebra: the grading is by the number of nodes, and
  $(\CK)_0$ is spanned by the algebra unit, the empty forest.  Therefore, by 
  general principles (see for 
  example \cite{Figueroa-GraciaBondia:0408145}), it acquires
  an antipode and becomes a Hopf algebra.
\end{blanko}

\begin{blanko}{Combinatorial Dyson--Schwinger equations.}\label{BK-DSE}
  The {\em combinatorial Dyson--Schwinger equations} of Bergbauer and
  Kreimer~\cite{Bergbauer-Kreimer:0506190} refer to an ambient combinatorial
  Hopf algebra $\HH$ and a collection of Hochschild $1$-cocycles.  By {\em
  Hochschild $1$-cocycle} is meant a linear operator $B_+$ satisfying the
  equation
  $$
  \Delta \circ B_+ = \big((\Id\tensor B_+) + (B_+ \!\tensor \eta\epsilon)\big) \circ \Delta  
  $$
  where $\epsilon$ is the counit, and $\eta$ is the algebra unit.
  (Note that this is not the standard Hochschild cohomology: $\HH$ is considered
  a bicomodule over itself via the identity action on the left and via $\eta \circ 
  \epsilon$ on the right (see Moerdijk~\cite{Moerdjik:9907010}).)
  
  The general form of combinatorial Dyson--Schwinger equations considered by Bergbauer and
  Kreimer is:
  \begin{equation}\label{eq:BK-DSE}
    X = 1 + \sum_{n\geq 1} w_n \alpha^n \; B^n_+ (X^{n+1})  .
  \end{equation}
  Here the $B^n_+$ are a sequence of $1$-cocycles, $w_n$ are scalars, and the
  parameter $\alpha$ is a coupling constant.  The solution $X$ will be a formal
  series, an element in $\HH[[\alpha]]$.  By making the Ansatz $X= \sum_{k\geq
  0} c_k \alpha^k$, plugging it into the equation, and solving for powers of
  $\alpha$, it is easy to see that there is a unique solution, which can be
  calculated explicitly up to any given order, as exemplified below.
\end{blanko}

\begin{theorem}[Bergbauer and Kreimer~\cite{Bergbauer-Kreimer:0506190}]
  \label{thm:BK}
  These $c_k$ span a Hopf sub-algebra of $\HH$, isomorphic to the Fa\`a di Bruno
  Hopf algebra.
\end{theorem}
In Quantum Field Theory, $\HH$ is a Hopf algebra of Feynman graphs.  The
importance of this theorem is that while $\HH$ itself is inherently of
perturbative nature, the
solution $X$, and hence the Fa\`a di Bruno Hopf sub-algebra spanned by it,
point towards non-perturbative interpretations, as explained by 
Kreimer~\cite{Kreimer:0609004}.

For many purposes, such
as the original purpose of organising the counter-terms in BPHZ 
renormalisation~\cite{Kreimer:9707029}, one can reduce to $\CK$, the Hopf algebra
of combinatorial trees.  In this case, there is only one $B_+$-operator, namely
the one that takes as input a forest and grafts the trees in it onto a new root
node to produce a single tree.  The following few examples refer to this Hopf
algebra.

\begin{blanko}{Example: a quadratic Dyson--Schwinger 
  equation~\cite{Bergbauer-Kreimer:0506190}.}\label{ex:binary-DSE}
  In the example
  $$
  X=1+\alpha B_+(X^2) ,
  $$
  by writing down the Ansatz $X = \sum_{k\geq 0} c_k \alpha^k$, 
  plugging it into the equation, and solving for powers of $\alpha$,
  one readily finds
  $$
  c_0 = 1, 
  \ c_1 = \raisebox{1pt}{\ctreeone} , 
  \ c_2 = 2 \raisebox{-2pt}{\ctreetwo} ,
  \ c_3 = 4 \raisebox{-4pt}{\ctreethreeL} + \raisebox{-2pt}{\ctreethreeV} ,
  \ c_4 = 8 \raisebox{-8pt}{\ctreefourL} 
  + 2 \!\!\raisebox{-4pt}{\ctreefourY} + 4 \raisebox{-4pt}{\ctreefourVL} ,
  \quad \text{etc.}
  $$
\end{blanko}

\begin{blanko}{Example with more complicated trees.}\label{ex:infinite-DSE}
  For the following `infinite' example,
  $$
  X = 1 + \sum_{n\geq 1} \alpha^n \; B_+ (X^{n+1}) ,
  $$
  with $X= \sum_{k\geq 0} c_k \alpha^k$, one finds
$$
c_0 = 1, \ \
  c_1 = \raisebox{1pt}{\ctreeone} , \ \
  c_2 = 2 \raisebox{-2pt}{\ctreetwo} + \raisebox{1pt}{\ctreeone}, \ \
  c_3 = 4 \raisebox{-4pt}{\ctreethreeL} + \raisebox{-2pt}{\ctreethreeV}
  + 5 \raisebox{-2pt}{\ctreetwo} + \raisebox{1pt}{\ctreeone}, 
  \phantom{xxxxxxxxxxxxxx}
$$
$$\phantom{xxxxxxxxxxxxxx} c_4 = 8 \raisebox{-8pt}{\ctreefourL} 
  + 2 \raisebox{-4pt}{\ctreefourY} + 4 \raisebox{-4pt}{\ctreefourVL} 
  + 16 \raisebox{-4pt}{\ctreethreeL} + 5 \raisebox{-2pt}{\ctreethreeV}
  + 9 \raisebox{-2pt}{\ctreetwo} + \raisebox{1pt}{\ctreeone},
  \quad \text{etc.}
$$
\end{blanko}

\begin{blanko}{Example with denominators.}\label{ex:infinite-symmetric-DSE}
  We consider finally an example where the constants $w_n$ are symmetry factors:
  $$
  X = 1 + \sum_{n\geq 1} \frac{1}{(n+1)!}\alpha^n \; B_+ (X^{n+1}) .
  $$
  With $X= \sum_{k\geq 0} c_k \alpha^k$ again, one finds this time
  $$
  c_0 = 1, \ \
  c_1 = \brok{1}{2}\raisebox{1pt}{\ctreeone} , \ \
  c_2 = \brok{1}{2} \raisebox{-2pt}{\ctreetwo} + \brok{1}{6}\raisebox{1pt}{\ctreeone}, \ \
  c_3 = \brok{1}{2} \raisebox{-4pt}{\ctreethreeL} + \brok{1}{8}\raisebox{-2pt}{\ctreethreeV}
  + \brok{5}{12} \raisebox{-2pt}{\ctreetwo} + \brok{1}{24}\raisebox{1pt}{\ctreeone},
  \phantom{xxxxxxxxxxxxx}
  $$
  $$\phantom{xxxxxxxxxxxxxx}
  c_4 = \brok{1}{2} \raisebox{-8pt}{\ctreefourL} 
  + \brok{1}{8} \raisebox{-4pt}{\ctreefourY} + \brok{1}{4} \raisebox{-4pt}{\ctreefourVL} 
  + \brok{2}{3} \raisebox{-4pt}{\ctreethreeL} + \brok{5}{24} \raisebox{-2pt}{\ctreethreeV}
  + \brok{5}{24} \raisebox{-2pt}{\ctreetwo} + \brok{1}{120}\raisebox{1pt}{\ctreeone},
  \quad \text{etc.}
  $$
\end{blanko}

While these expansions are straightforward to calculate, the pattern governing
them may not be obvious.  What is really going on will only become clear once we
pass to the setting of polynomial functors (see
Example~\ref{ex:stable-planar} for the case without symmetry factors, and
Example~\ref{ex:stable} for the
case with symmetry factors).

Foissy~\cite{Foissy:0707.1204} studies a different form of
combinatorial Dyson--Schwinger equations, whose solutions do not automatically
generate Hopf sub-algebras. The main difference is the role of the coupling
  constant and the grading it defines: in the Foissy equation, the power of
  $\alpha$ counts the number of $B_+$-operators, whereas in the
  Bergbauer--Kreimer equation it relates to the power of $X$ inside the
  $B_+$-operator, which is the operadic grading, as we shall see. 
We shall briefly return to Foissy's equations in 
Section~\ref{sec:Foissy}.

\section{Polynomial functors in one variable; basic aspects of the theory}
\label{sec:1var}

The notion of polynomial functor has origins in Topology, Representation Theory,
Combinatorics, Logic, and Computer Science, but the task of unifying these
developments has only recently begun~\cite{Gambino-Kock:0906.4931}.
For the present purposes, the natural level of generality is that of
many-variable polynomial functors over groupoids~\cite{Kock:MFPS28}.  This is
needed for example to handle Feynman graphs encoded as trees.  However, most of
the features of the theory developed here can be appreciated already in the setting of
one-variable polynomial functors over sets, and this requires almost no
background in Category Theory.

For the benefit of the less categorically inclined reader, we go through this
easier case first, to explain the main ideas, before turning to the general
setting in Section~\ref{sec:groupoid-DSE}, where the more technical aspects will be treated.

\begin{blanko}{Categories and functors.}
  Only very little category theory is needed for this section.
  For further background, see Leinster~\cite{Leinster:BCT} for a short and
  concise introduction, and Spivak~\cite{Spivak:CTS} for an account addressed at
  non-mathematicians.
  
  A {\em category} has objects and arrows (morphisms), and arrows can be composed.
  The primordial example is the category of sets, denoted $\Set$, where the objects
  are sets and the arrows are maps of sets.  Another example is the category
  $\kat{Vect}$ of vector spaces and linear maps.  A {\em functor} is a `morphism of
  categories', i.e.~sends objects to objects and arrows to arrows, in such
  a way as to preserve composition.  For example there is a functor $F: \Set \to
  \kat{Vect}$, sending a set $S$ to the vector space spanned by $S$, and sending
  a set map $f:S \to T$ to the linear map induced by its value on basis vectors.
\end{blanko}

\begin{blanko}{Sets.}
  The category of sets is the origin of elementary arithmetic: the {\em sum} of
  two sets $A$ and $B$ it their disjoint union, written with a sum sign
  $A+B$ to emphasise the additive nature of this operation: the cardinality of
  $A+B$ is the cardinality of $A$ plus the cardinality of $B$.  Furthermore,
  the sum has the universal property that to specify a map from $A+B$ to some 
  set $X$ is the same thing as specifying one map $A \to X$ and another map 
  $B\to X$.  Similarly, the {\em product} of $A$ and $B$ is their cartesian
  product.  The cardinality of $A\times B$ is of course the cardinality of $A$
  times the cardinality of $B$.  Finally we use the {\em exponential} notation
  $B^A$ for the set of maps
  from $A$ to $B$; this notation
  (standard in category theory) is justified by the fact that if $B$ is an
  $n$-element set, and $A$ is a $k$-element set, then $B^A$ is an $n^k$-element
  set. 
\end{blanko}

\begin{blanko}{Polynomial functors in one variable.}
  In their simplest manifestation, polynomial functors are endo\-functors of the
  category of sets built from arbitrary sums and products (and hence also
  constants and exponentiation).  A standard reference for polynomial functors
  is \cite{Gambino-Kock:0906.4931}; the
  manuscript~\cite{Kock:NotesOnPolynomialFunctors} aims at eventually becoming a
  unified reference.

  Given a map of sets $p: E\to B$, we define the associated {\em polynomial functor} 
  in one variable to be the functor
  \begin{eqnarray}
    P:\Set & \longrightarrow & \Set \notag  \\
    X & \longmapsto & \sum_{b\in B} X ^{E_b} .\label{eq:poly1var}
  \end{eqnarray}
  In the formula, the sum sign denotes disjoint union of sets, and $E_b =
  p^{-1}(b)$ denotes the inverse image of an element $b\in B$, also called the 
  fibre over $b$.
  We see that the role played by the map $p:E \to B$ is to deliver a
  family of sets indexed by $B$, namely
  \begin{equation}
  \label{eq:setfamily}
  ( E_b \mid b\in B).
  \end{equation}
  We say that $E_b$ is the {\em arity} of $b$.
  Note that $B$ may be an infinite set; then the sum in \eqref{eq:poly1var} is
  accordingly an infinite sum (so in a sense polynomial functors are more like
  power series than like polynomials).  On the other hand, we shall always require 
  that
  the map $p$ has finite fibres.  Such polynomial functors are called {\em 
  finitary}.

  To say that $P$ is a {\em functor} means that it operates not just on sets but 
  also on maps: given a map of sets $a:X \to Y$, there is induced a map
  $$
  \sum_{b\in B} X ^{E_b} \to \sum_{b\in B} Y ^{E_b}
  $$
  termwise given by
  \begin{eqnarray*}
    X^{E_b} & \longrightarrow & Y^{E_b}  \\
    f & \longmapsto & a\circ f.
  \end{eqnarray*}
\end{blanko}

\begin{blanko}{Polynomial fixpoint equations.}
  The abstract combinatorial `Dyson--Schwinger equations' we shall consider here
  are equations of the form
  \begin{equation}\label{eq:polyfix}
  X \isopilback 1 +P(X) ,
  \end{equation}
  where $P$ is a polynomial functor, and $1$ denotes a singleton set.  This is
  an equation of sets, and to solve it means to find a set $X$ together with a
  specific bijection with $1+P(X)$, as indicated by the symbol $\isopilback$.
  In fact, we are not satisfied with finding {\em some} solution; we want the
  {\em best} solution, the {\em least fixpoint}.  Making this precise requires a
  few more notions of category theory:
\end{blanko}
  
\begin{blanko}{Initial objects.}
  An object $I$ in a category $\CC$ is called {\em initial} if for every object
  $C$ there is a unique morphism in $\CC$ from $I$ to $C$.  It is easy to show
  that an initial object, if it exists, is unique (up to isomorphism).  For
  example, the category of sets has an initial object, namely the empty set
  $\emptyset$, since for any set $X$ there is a unique set map $\emptyset \to X$.
\end{blanko}

\begin{blanko}{$P$-algebras and initial algebras.}
  A {\em $P$-algebra} is by definition a pair $(A,a)$ where $A$ is a set, and
  $a: P(A) \to A$ is a set map.  A {\em homomorphism} of $P$-algebras from $(A,a)$ to 
  $(B,b)$ is a set map $f:A\to B$ compatible with the structure maps $a$ and $b$,
  i.e.~such that this square commutes:
  \begin{equation}\label{eq:P-alg map}
    \xymatrix{
     P(A) \ar[r]^-a\ar[d]_{P(f)} & A \ar[d]^f \\
     P(B) \ar[r]_-b & B   .
  }
  \end{equation}
  Note that the functoriality of $P$ is necessary even to be able to state this
  compatibility: we need to be able to evaluate $P$ not just on sets but also
  on maps.  Altogether, there is a category $P\alg$ of $P$-algebras
  and $P$-algebra homomorphisms.
    
  {\em Lambek's lemma}~\cite{Lambek:fixpoint}
  says that if the category of $P$-algebras has an initial
  object $(A,a)$, then the structure map $a$ is invertible.  (This is not a
  difficult result.)  This is precisely to say that an initial $P$-algebra
  $(A,a)$ is a solution to the equation $X \isopilback P(X)$: the underlying set
  $A$ is $X$, and the structure map $a$ is the required bijection.  Initiality
  is the technical condition that justifies referring to this as the {\em least
  fixpoint}.
  
  We shall later (\ref{bla:Foissy}) come back to the equation $X \isopilback P(X)$.
  Right now, the equation we wanted to solve is rather
  $$
  X\isopilback 1+P(X),
  $$
  so what we are looking for is the initial $(1+P)$-algebra rather than the initial
  $P$-algebra.   (Here $1$ denotes the constant functor $1$, and the sum is 
  pointwise sum of $\Set$-valued functors.)
\end{blanko}

\begin{theorem}
  If $P$ is a polynomial functor, then the fixpoint equation
  $$
  X\isopilback 1+P(X)
  $$ 
  has a least solution, that is, the category of $(1+P)$-algebras has an initial
  object.  Assuming $P$ is finitary, this solution is the set of (isomorphism
  classes of) $P$-trees, now to be defined.
\end{theorem}

The existence of initial algebras for polynomial functors is a classical result
in Category Theory, going back to Lambek in the late 1960s~\cite{Lambek:fixpoint}.
The initial $(1+P)$-algebra is also the set of operations for the free monad
on $P$, somewhat explaining the importance of the equation $X\isopilback 1+P(X)$
over $X\isopilback P(X)$.
The explicit characterisation of the solution in terms of $P$-trees
is from \cite{Kock:0807}.

\begin{blanko}{Operadic trees.}\label{operadictrees}
  By {\em operadic trees} we mean rooted trees admitting open-ended edges
  (leaves and root), such as the following:
  $$
    \begin{texdraw}
  \linewd 0.8 \footnotesize
  \move (-50 0)
  \bsegment
    \move (0 0) \lvec (0 20)
  \esegment
  
  \move (0 0)
  \bsegment
    \move (0 0) \lvec (0 10) \onedot
  \esegment
  
  \move (50 0)
  \bsegment
    \move (0 0) \lvec (0 26)
    \move (0 13) \onedot
  \esegment
  
  \move (105 0)
  \bsegment
    \move (0 0) \lvec (0 10) \onedot
    \lvec (-5 25) \onedot \lvec (-10 40)
    \move (0 10) \lvec (-14 21) \onedot
    \move (0 10) \lvec (5 25) \onedot \lvec (0 40)
    \move (5 25) \lvec (10 40)
    \move (0 10) \lvec (23 36)
  \esegment
  \end{texdraw}
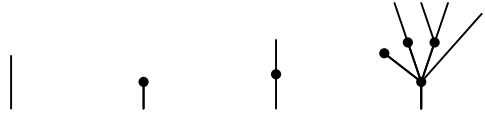
  $$
  A formal definition will be given in \ref{polytree-def}, exploiting the many-variable
  setting.  Trees of this kind are called operadic because each node is to be
  thought of as an operation, with its incoming edges (reading from top to
  bottom) as input slots and its outgoing edge as output slot.  Note the
  difference between a leaf (an open-ended edge) and a nullary node.  In 
  contrast to operadic trees, we refer to the usual trees without open-ended 
  edges as {\em combinatorial trees}.
\end{blanko}

\begin{blanko}{$P$-trees.}\label{P-trees-naive}
  (We shall come to a more abstract description in \ref{Ptree} and 
  \ref{PtreeGr}.)
  Let $P : \Set\to\Set$ be a finitary polynomial functor represented by a set map
  $p:E\to B$.
  A {\em $P$-tree} is an operadic tree with nodes decorated by elements in $B$,
  and for each node $x$ decorated by $b$ a specified bijection between the
  incoming edges of $x$ and the set $E_b$.  In other words, each node is
  decorated with an operation of matching arity.
\end{blanko}

\begin{blanko}{Core.}\label{core}
  The {\em core} of a $P$-tree is the combinatorial tree constituted by its
  inner edges.  So to obtain it, forget all decorations, and shave off leaf
  edges and root edge.  (The word `core' is used for something else in 
  \cite{Kreimer:0902.1223}, but confusion should not be likely.)
\end{blanko}

\begin{blanko}{Example: binary trees.}\label{ex:binary}
  Consider the polynomial functor $P$ defined by the set map $p:
  \{\mathtt{left},\mathtt{right}\} \to 1$.  It is the functor
  \begin{eqnarray*}
    \Set & \longrightarrow & \Set  \\
    X & \longmapsto & X^2 .
  \end{eqnarray*}
  For this $P$, a $P$-tree is precisely a (planar) binary tree.
  Indeed, since in this case the set $B$ is just singleton,
  to $P$-decorate a tree amounts to specifying for each
  node a bijection between the set of incoming edges and the set 
  $\{\mathtt{left},\mathtt{right}\}$.  For this bijection to be possible,
  each node must have precisely two incoming edges, and the bijection says
  which is the left branch and which is the right.

  The corresponding fixpoint equation $X\isopilback 1+P(X)$ is now
  $$
  X \isopilback 1 + X^2 ,
  $$
  and
  the theorem thus says that the solution, the initial $(1+P)$-algebra, is the set
  of planar binary trees.  Indeed, the fixpoint equation can be read as saying:
  a planar binary tree is either the trivial tree, or it is given by a pair of
  planar binary trees.  This is precisely the recursive characterisation of
  binary trees.
  Here are the first few elements:

  \begin{equation}\label{eq:binary-trees}\begin{texdraw}\setunitscale 0.8
    \move (0 0) 
    \bsegment
      \rlvec (0 20) 
    \esegment
    
    \move (40 0) 
    \bsegment
      \move (0 0)
      \lvec (0 10) \onedot \rlvec (-5 10)
      \move (0 10) \rlvec (5 10)
    \esegment
    
    \move (85 0)   
    \bsegment
      \move (0 0)
	\lvec (0 10) \onedot \rlvec (-10 20)
        \move (0 10) \rlvec (10 20)
        \move (-5 20) \onedot \rlvec (5 10)
    \esegment

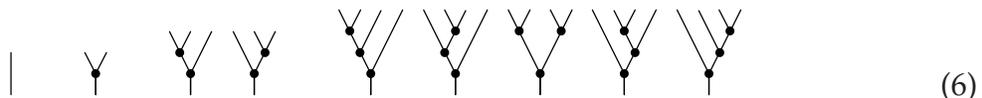
    \move (115 0)   
    \bsegment
      \move (0 0)
	\lvec (0 10) \onedot \rlvec (-10 20)
        \move (0 10) \rlvec (10 20)
        \move (5 20) \onedot \rlvec (-5 10)
    \esegment
    
    \move (170 0)
    \bsegment
      \move (0 0) 
      \bsegment
	\move (0 0) 
	\lvec (0 10) \onedot \rlvec (-15 30)
	\move (0 10) \rlvec (15 30)
	\move (-5 20) \onedot \rlvec (10 20)
	\move (-10 30) \onedot \rlvec (5 10)
      \esegment
     
      \move (40 0) 
      \bsegment
	\move (0 0) 
	\lvec (0 10) \onedot \rlvec (-15 30)
	\move (0 10) \rlvec (15 30)
	\move (-5 20) \onedot \rlvec (10 20)
	\move (0 30) \onedot \rlvec (-5 10)
      \esegment
   
      \move (80 0) 
      \bsegment
	\move (0 0) 
	\lvec (0 10) \onedot \rlvec (-15 30)
	\move (0 10) \rlvec (15 30)
	\move (-10 30) \onedot \rlvec (5 10)
	\move (10 30) \onedot \rlvec (-5 10)
      \esegment
     
      \move (120 0) 
      \bsegment
	\move (0 0) 
	\lvec (0 10) \onedot \rlvec (-15 30)
	\move (0 10) \rlvec (15 30)
	\move (5 20) \onedot \rlvec (-10 20)
	\move (0 30) \onedot \rlvec (5 10)
      \esegment

      \move (160 0) 
      \bsegment
	\move (0 0) 
	\lvec (0 10) \onedot \rlvec (-15 30)
	\move (0 10) \rlvec (15 30)
	\move (5 20) \onedot \rlvec (-10 20)
	\move (10 30) \onedot \rlvec (-5 10)
      \esegment
      
    \esegment
  \end{texdraw}\end{equation}

  Sorting by the number of leaves and
  taking core sends these binary trees to combinatorial trees which are
  sub-binary (i.e.~have at most two incoming edges at each node), and where the
  planar structure has been lost.  We see that the coefficients $c_k$ appearing
  in the solution of the quadratic Dyson--Schwinger equation \ref{ex:binary-DSE}
  are precisely the numbers of binary trees with a given core.  (This
  interpretation of the coefficients $c_k$ was given already by Bergbauer and
  Kreimer~\cite{Bergbauer-Kreimer:0506190}, in some form.)
\end{blanko}

\begin{blanko}{The structure bijection as $B_+$-operator.}
  We continue the example of binary trees.
  Ignoring for a moment that the structure map for the $(1+P)$-algebra $X$
  is a bijection, it is first of all a map 
  $$
  X \leftarrow 1 + X^2 .
  $$
  By the universal property of the sum, this map consists of two maps:
  $$
  X \leftarrow 1 \qquad \text{ and } \qquad X \leftarrow X^2 .
  $$
  The first picks out the trivial tree.  The second associates to each
  pair of trees a single tree; by unravelling the bijection this single 
  tree is seen to be precisely obtained by grafting the
  two trees onto a one-node binary tree.  We see that the abstract bijection
  of sets realising the solution is precisely the $B_+$-operator.
  (A more detailed argument is given in the groupoid case in the proof of 
  Theorem~\ref{thm:hoinit-gr}.)
\end{blanko}

\begin{blanko}{Example: planar trees.}\label{ex:planar}
  Take $P(X) = X^0 + X^1 + X^2 + X^3 + \cdots$ This is the list
  endofunctor, sending a set $X$ to the set of lists of elements in $X$.
  This is a polynomial functor, represented by 
  $$
  \N'\to\N,$$
  where $\N$ is (isoclasses of) finite ordered sets (i.e.~the natural numbers)
  and $\N'$ is (isoclasses of) finite pointed ordered sets.
  The fibre over any $n$-element set is that set itself, hence accounting for
  the exponents $X^n$ in the expression for $P$.
  Then $P$-trees are planar trees.
  The fixpoint equation 
  $$
  X \isopilback 1 + X^0 + X^1 + X^2 + X^3 + \cdots
  $$
  says that a planar tree is either the trivial tree or a list of planar trees.

  Since we allow nullary and unary nodes (corresponding to the terms $X^0$ and
  $X^1$ in the polynomial functor), for each fixed number of leaves, there are
  infinitely many trees.  For the sake of comparison with \ref{BK-DSE}, it is
  interesting to tweak this functor a little bit, to avoid this infinity:
\end{blanko}
  
\begin{blanko}{Example: stable planar trees.}\label{ex:stable-planar}
  Consider instead the polynomial functor 
  $$P(X) = X^2 + X^3 +X^4 + \cdots$$
  for which $P$-trees are {\em stable} planar trees, meaning without nullary or 
  unary nodes.  The exclusion of nullary and 
  unary nodes implies that now for each fixed number $k$ there are only
  finitely many trees with $k$ leaves. 
  These are the 
  Hipparchus--Schr\"oder numbers, $1, 1, 3, 11, 45, 197, 903,\ldots$
  Here are pictures of all the stable
  planar trees with up to $4$ leaves:
   
    \begin{equation}\label{eq:stable-planar-trees}\begin{texdraw}
      \setunitscale 0.8
    \move (0 0) 
    \bsegment
      \rlvec (0 20) 
    \esegment
    
    \move (80 0) 
    \bsegment
      \move (0 0)
      \lvec (0 10) \onedot \rlvec (-5 10)
      \move (0 10) \rlvec (5 10)
    \esegment
    
    \move (160 -5)   
    \bsegment
      \move (0 0)
          \bsegment
	    \move (0 0)
	    \lvec (0 10) \onedot \rlvec (-10 20)
	    \move (0 10) \rlvec (10 20)
	    \move (-5 20) \onedot \rlvec (5 10)
	  \esegment

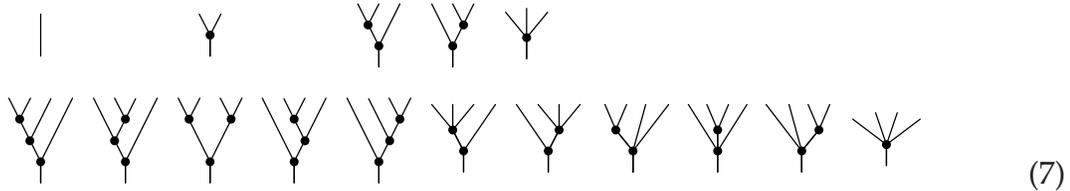
	  \move (35 0)
	  \bsegment
	    \move (0 0)
	    \lvec (0 10) \onedot \rlvec (-10 20)
	    \move (0 10) \rlvec (10 20)
	    \move (5 20) \onedot \rlvec (-5 10)
	  \esegment
	  \move (70 0)
	  \bsegment
	    \move (0 4)
	    \lvec (0 14) \onedot \rlvec (-10 12)
	    \move (0 14) \rlvec (10 12)
	    \move (0 14) \rlvec (0 14)
	  \esegment
    \esegment
    
    \move (0 -60)
    \bsegment
      \move (0 0) 
      \bsegment
	\move (0 0) 
	\lvec (0 10) \onedot \rlvec (-15 30)
	\move (0 10) \rlvec (15 30)
	\move (-5 20) \onedot \rlvec (10 20)
	\move (-10 30) \onedot \rlvec (5 10)
      \esegment
     
      \move (40 0) 
      \bsegment
	\move (0 0) 
	\lvec (0 10) \onedot \rlvec (-15 30)
	\move (0 10) \rlvec (15 30)
	\move (-5 20) \onedot \rlvec (10 20)
	\move (0 30) \onedot \rlvec (-5 10)
      \esegment
   
      \move (80 0) 
      \bsegment
	\move (0 0) 
	\lvec (0 10) \onedot \rlvec (-15 30)
	\move (0 10) \rlvec (15 30)
	\move (-10 30) \onedot \rlvec (5 10)
	\move (10 30) \onedot \rlvec (-5 10)
      \esegment
     
      \move (120 0) 
      \bsegment
	\move (0 0) 
	\lvec (0 10) \onedot \rlvec (-15 30)
	\move (0 10) \rlvec (15 30)
	\move (5 20) \onedot \rlvec (-10 20)
	\move (0 30) \onedot \rlvec (5 10)
      \esegment
       
      \move (160 0) 
      \bsegment
	\move (0 0) 
	\lvec (0 10) \onedot \rlvec (-15 30)
	\move (0 10) \rlvec (15 30)
	\move (5 20) \onedot \rlvec (-10 20)
	\move (10 30) \onedot \rlvec (-5 10)
      \esegment

      \move (200 0) 
      \bsegment
	\move (0 5) 
	\lvec (0 15) \onedot \rlvec (15 22)
	\move (0 15) \lvec (-5 25) 
	\onedot \rlvec (-10 12)
	\move (-5 25) \rlvec (10 12)
	\move (-5 25) \rlvec (0 12)
      \esegment

      \move (240 0) 
      \bsegment
	\move (0 5) 
	\lvec (0 15) \onedot \rlvec (-15 22)
	\move (0 15) \lvec (5 25) 
	\onedot \rlvec (-10 12)
	\move (5 25) \rlvec (10 12)
	\move (5 25) \rlvec (0 12)
      \esegment

      \move (280 0) 
      \bsegment
	\move (0 5) 
	\lvec (0 15) \onedot \lvec (-8 25) \onedot
	\rlvec (-5 12) \move (-8 25)
	\rlvec (5 12)
	\move (0 15) \lvec (6 37) 
	\move (0 15) \lvec (17 37) 
      \esegment
      
      \move (320 0) 
      \bsegment
	\move (0 5) 
	\lvec (0 15) \onedot \lvec (-14 37)
	\move (0 15) \lvec (14 37)
	\move (0 15) \lvec (0 25) \onedot
	\rlvec (-5 12) \move (0 25)
	\rlvec (5 12)
      \esegment
      
      \move (360 0) 
      \bsegment
	\move (0 5) 
	\lvec (0 15) \onedot \lvec (8 25) \onedot
	\rlvec (5 12) \move (8 25)
	\rlvec (-5 12)
	\move (0 15) \lvec (-6 37) 
	\move (0 15) \lvec (-17 37) 
      \esegment

      \move (400 0) 
      \bsegment
	\move (0 8) 
	\lvec (0 18) \onedot 
	\lvec (-16 30)
	\move ( 0 18) \lvec (-5 33)
	\move ( 0 18) \lvec (5 33)
	\move ( 0 18) \lvec (16 30)
      \esegment

    \esegment
    
    \setunitscale 1
  \end{texdraw}\end{equation}  
  Sorting by the number of leaves and
  taking core yields precisely the 
  combinatorial trees found in the solution to the Dyson--Schwinger equation
  \eqref{ex:infinite-DSE}, and again the coefficients $c_k$ in the solution
  are precisely the numbers of trees with $k+1$ leaves and with a given core.
\end{blanko}

\begin{blanko}{Multiple $B_+$-operators.}
  In the previous two examples, the structure bijection is constituted by
  maps $X \leftarrow 1$ and
  $$
  X \leftarrow X^n
  $$
  for all $n$ (or all $n\geq2$).  What each of these maps does is to take
  an $n$-tuple of planar trees and return a single tree obtained by grafting
  all these trees onto an $n$-ary one-node planar tree.  A first
  $B_+$-operator interpretation of the fixpoint equation
  in this case is therefore that there is one
  $B_+$-operator for each $n$, 
  corresponding to grafting onto an $n$-ary node,  
  and that each of these operators only accepts
  a fixed-size tuple of trees as input, namely an $n$-tuple.
  In Equation~\eqref{ex:infinite-DSE},
  $$
  X = 1 + \sum_{n\geq 1} \alpha^n \; B_+ (X^{n+1}) ,
  $$  
  there is reference to only one $B_+$-operator, but there is nevertheless some
  distinction implicit in the equation, because the coupling constant $\alpha$
  keeps track of how many input trees the operator takes: the $n$th power of
  $\alpha$ corresponds to the $B_+$-operator with $n+1$ inputs,
  and hence altogether, in the solution, power $\alpha^n$ corresponds to
  trees with $n+1$ leaves.
\end{blanko}

\begin{blanko}{Bialgebra of $P$-trees ($1$-variable version).}\label{treebialg1var}
  For any finitary polynomial endofunctor $P$, the set of (isoclasses of)
  $P$-trees form a Connes--Kreimer-style bialgebra, first studied in~\cite{Kock:1109.5785}
   and subsumed into the general framework of incidence coalgebras of 
   decomposition spaces in \cite{Galvez-Kock-Tonks:1512.07573}; see 
   \cite{Galvez-Kock-Tonks:1612.09225} for detailed discussion.
   The idea is the same as for the Connes--Kreimer Hopf algebra, to sum over 
   cuts, but there is the important difference that
   when cutting edges, they are really cut, not removed, as
   exemplified by
  
  \vspace{-4pt}
  
  \begin{center}\begin{texdraw}
    \htext (0 0) {$\Delta($}
   \rmove (12 0) \twotree
     \rmove (15 0)
    \htext  {$) \ \ = $}
        \rmove (23 0) \zerotree
    \rmove (3.5 0) \zerotree
    \rmove (3.5 0) \zerotree

          \rmove (11 0) \htext {$\tensor$}
   \rmove (14 0) \twotree

	            \rmove (18 0) \htext {$+$}

    \rmove (20 0) \onetree  \rmove (5 0) \zerotree 
    
    \rmove (12 0) \htext {$\tensor$}
   \rmove (12 0) \onetree
   
      \rmove (19 0) \htext {$+$}

    \rmove (21 0) \twotree
    
    \rmove (13 0) \htext {$\tensor$}
   \rmove (10 0) \zerotree
  \end{texdraw}\end{center}
  This is an essential point, as otherwise the decorations would be spoiled:
  removing an edge instead of cutting it would break the given arity bijections
  in the decorations (not rendered in the drawing).  The counit $\varepsilon$
  sends nodeless trees to $1$, and all other trees to $0$.  Note that
  this bialgebra is not connected: the degree-zero piece is spanned by the
  nodeless trees and forests.
\end{blanko}

Taking core constitutes a bialgebra homomorphism from the bialgebra of $P$-trees
to the Connes--Kreimer Hopf algebra (cf.~\ref{core-bialghomo} for this result in
the general setting).

\begin{blanko}{$B_+$-operators inside the bialgebra of $P$-trees.}
  In the case of planar (and stable planar) trees we already saw
  how the structure bijection of the solution to the fixpoint equation
  can be interpreted in terms of $B_+$-operators.  This generalises to
  arbitrary polynomial endofunctors $P$, yielding linear operators on
  the bialgebra of $P$-trees.
  Namely, for each $b\in B$ there is a $B_+$-operator corresponding to 
  the $b$-summand of the structure map
  $$
  X \isopilback 1+\sum_{b\in B} X^{E_b} ;
  $$
  this map takes an $E_b$-tuple of $P$-trees and associates to it
  a single $P$-tree obtained by grafting all those $P$-trees onto
  the $P$-corolla labelled by $b$.  Each of these extends to a linear map.
  All these maps are supplemented by the map $X \leftarrow 1$ which 
  singles out the trivial $P$-tree. 
  
  The very bijection altogether says that every non-trivial $P$-tree is in the
  image of (precisely) one $B_+$-operator.  This fact is important in Quantum
  Field Theory as it ensures locality of counter terms
  (Bergbauer--Kreimer~\cite{Bergbauer-Kreimer:0506190}, \S1.6).  Bergbauer and
  Kreimer derive this fact from the $1$-cocycle condition of the
  $B_+$-operators.  In the present setting it is a
  direct consequence already of the fixpoint property.  In fact, at the natural
  $\Set$-level, the $B_+$-operators are {\em not} $1$-cocycles, as we shall see
  in \ref{notHH}.
 
  Some further subtleties related to symmetries deserve attention here.
  The input to the $B_+$-operator should be an abstract forest (a $P$-forest, 
  of course), not an $E_b$-indexed family of $P$-trees.  The final version of
  the $B_+$-operator takes an abstract $P$-forest $F$ and turns it into an 
  $E_b$-indexed family of $P$-trees in all possible ways, compatible with the 
  colours (see \ref{B+} for precision).  This introduces a factor
  $$
  \Bij(\operatorname{leaves}(b), \operatorname{roots}(F)) .
  $$
  But the group $\Aut(b)$ acts naturally (and freely) on this set, and it is the
  quotient of this action that matters, as we are talking isoclasses of trees at
  this point.  In the case of naked trees (which will be covered by the notion 
  of $P$-tree once we upgrade to groupoids~(\ref{ex:naked})), the automorphism group of an
  $n$-corolla $b$ is the full symmetry group $\mathfrak S_n$, so the action is
  also transitive in this case, and altogether the quotient is a singleton set,
  so no symmetry factors arise in this case.  In the case of planar trees, 
  $\Aut(b)$ is trivial and there would always be a factor $n!$.  In the 
  literature (notably Foissy~\cite{Foissy:0707.1204}), this is sometimes circumvented
  by considering a noncommutative algebra, consisting of {\em planar} forests 
  (i.e.~lists of planar trees).  In quantum field theory, the symmetry issues
  are considerably more complicated and cannot be bypassed, as is well known
  (see for example \cite{Kreimer:0509135} for detailed discussion and example 
  computations).
  We shall see that once the polynomial machinery is upgraded to the groupoid
  setting, all symmetry issues take care of themselves completely transparently.
\end{blanko}

\begin{blanko}{Green functions.}
  The abstract Dyson--Schwinger equation at the $\Set$-level can now 
  be internalised to the bialgebra.  While the set solution has its own
  canonical structure (since each element in the set is a tree, with the 
  various attributes that can be associated to and read off a tree), 
  in the algebraic situation inside the bialgebra it is convenient to
  control these attributes by a formal parameter, the coupling 
  constant $\alpha$, which serves as a bookkeeping device to keep track of the
  number of leaves.
  
  The bialgebra has two natural gradings: one by the number of nodes and one
  by the number of leaves minus the number of roots.  Since every tree
  is inductively constructed by grafting to corollas, each step given by a 
  $B_+$-operator, we attach the power $\alpha^n$ to the corolla with $n+1$ 
  leaves, and hence altogether, trees or forests whose number of leaves minus
  number of roots equals $n$.  Note that a priori this allows for exponent
  $-1$ on trees (corresponding to trees without leaves), and arbitrary
  negative exponents on forests.  (It is not unreasonable in practice to 
  restrict to working
  with a $P$ without nullary operations so as to obtain a non-negative grading.)
  
  While in the traditional setting, as outlined in Section~\ref{sec:DSE}, the
  solution to the combinatorial Dyson--Schwinger equation is a series calculated
  inside a given Hopf algebra, at the natural $\Set$-level, the ambient
  bialgebra is designed specifically around the equation, and for this reason
  the solution inside it is tautological: the {\em Green function} $G$, which
  solves the internalised Dyson--Schwinger equation, is simply the formal series
  with one term for each generating element in the bialgebra, that is, each
  $P$-tree.  (At the $\Set$-level there are no symmetries, but we shall see in
  the groupoid case (Section~\ref{sec:groupoid-Green}) that symmetry factors
  naturally arise.)
  
  For very general reasons, expressible at the objective level, this
  Green function $G$ satisfies the Fa\`a di Bruno 
  formula~\cite{GalvezCarrillo-Kock-Tonks:1207.6404}:
  $G$ can be written as a sum 
  $$
  G = \sum_{n\geq 0} g_n,
  $$
  where $g_n$ is the sum of all trees with $n$ leaves,
  and now there is a Fa\`a di Bruno formula
  $$
  \Delta(G) = \sum_{n\geq 0} G^n \tensor g_n
  $$
  in the style of van Suijlekom~\cite{vanSuijlekom:0807}.  The point is that 
  the exponent $n$ on the left-hand tensor factor 
  counts $n$ trees, each with a root, precisely matching the 
  subscript $n$ in the right-hand tensor factor, which is the number of leaves in the 
  trees in $g_n$.  This kind of information cannot be seen in at the level of 
  combinatorial trees.
  
  (Calling this formula the Fa\`a di Bruno formula is justified by the fact that
  in the classical Fa\`a di Bruno bialgebra, dual to composition of power 
  series, the formula is equivalent to the classical Fa\`a di Bruno formula in 
  terms of Bell numbers.  This was discovered by Brouder, Frabetti and 
  Krattenthaler~\cite{Brouder-Frabetti-Krattenthaler:0406117}.  A very abstract
  version of the formula which englobes both the classical case and the case of
  trees \cite{GalvezCarrillo-Kock-Tonks:1207.6404} is established in 
  \cite{Kock-Weber:1609.03276}.)

  The formula shows that the homogeneous components $g_n$ span a sub-bialgebra
  isomorphic to the Fa\`a di Bruno bialgebra. We shall come back to these
  results (proved in \cite{GalvezCarrillo-Kock-Tonks:1207.6404}) in the groupoid setting
  of Section~\ref{sec:groupoid-Green}.
\end{blanko}

\section{Polynomial functors in many variables}
\label{sec:many}

\begin{blanko}{Polynomial functors in many variables.}
  We now pass to the multivariate case.  For purely mathematical reasons this is
  important because it allows trees themselves to be treated on equal footing
  with polynomial functors: trees {\em are} themselves polynomial functors in
  many variables, as we shall see in \ref{polytree-def}.  Importance comes also
  from the motivation in quantum field theory, since many-variable polynomial 
  functors $P$ are needed for the interpretation of Feynman graphs as $P$-trees,
  discussed in Section~\ref{sec:trees-graphs}.
  
  In the many-variable case, instead of just one variable input set $X$,
  we have an
  $I$-indexed family of input sets, and the output is not a single set but a
  $J$-indexed family of sets.  For the present purposes, we will always have
  $I=J$, so that the functor is still an endofunctor, but now on slices.
  An $I$-indexed family of sets $(X_i \mid i\in I)$ can conveniently be
  encoded as a single map of sets $f:X\to I$.  Then an individual family member
  $X_i$ is given by the fibre $X_i = f^{-1}(i)$.  The appropriate notion of
  morphism between $I$-indexed families of sets is given by the {\em slice
  category} $\Set_{/I}$: its objects are maps $X \to I$, and its morphisms are
  commutative triangles
  $$
  \xymatrix@C=3ex{
  X \ar[rr] \ar[rd] &  & Y \ar[ld] \\
  &I .&
  }$$

  Following \cite{Gambino-Kock:0906.4931}, a {\em polynomial} is a diagram of sets
  \begin{equation}\label{IEBJ}
  I \stackrel s \longleftarrow E \stackrel p \longrightarrow B \stackrel t 
  \longrightarrow J ,
  \end{equation}
  and
  the associated {\em polynomial functor} is given by the composite
  \begin{equation}\label{poly}
  \Set_{/I} \overset{s\upperstar}\longrightarrow
  \Set_{/E} \overset{p\lowerstar}\longrightarrow
  \Set_{/B} \overset{t\lowershriek}\longrightarrow
  \Set_{/J} \,,
  \end{equation}
  where $s\upperstar$ is pullback along $s$,
  $p\lowerstar$ is the right adjoint
  to pullback, 
  and $t\lowershriek$ is left adjoint to 
  pullback.
  For a map $f:B\to A$ we have the three explicit formulae
  \begin{align}
  f\upperstar (X_a\mid a\in A) &= ( X_{f(b)}\mid b\in B)
    \label{pbk}\\[3pt]
  f\lowershriek (Y_b \mid b\in B) &= (\underset{b\in B_a}{\textstyle{\sum}} Y_b \mid a\in A)
    \label{sum}\\
  f\lowerstar  (Y_b \mid b\in B) &= (\underset{b\in B_a}{\textstyle{\prod}} Y_b \mid a\in A) \,,
    \label{prod}
  \end{align}
  giving altogether the following formula for \eqref{poly}
  $$
  (X_i \mid i\in I) \longmapsto (\sum_{b\in B_j} \prod_{e\in E_b} X_{s(e)}  \mid 
  j\in J) ,
  $$
  which specialises to \eqref{eq:poly1var} when $I=J=1$.
\end{blanko}

\begin{blanko}{Many-variable fixpoint equations.}
  We can now formulate polynomial fixpoint equations in slice categories.
  For $P: \Set_{/I} \to \Set_{/I}$ a finitary polynomial endofunctor, given by
  $I \leftarrow E \to B \to I$, we consider the equation
  $$
  X \isopilback 1 + P(X) .
  $$
  The solution will be an object $X$ in the slice category $\Set_{/I}$, that is, 
  an $I$-indexed family of sets.  The symbol $1$ now denotes the terminal object
  in this slice category, which is the trivial family $I \to I$, the
  $I$-indexed family of singletons.  The general category theory that ensures the 
  existence of initial algebras in sets, also ensures the existence of initial 
  algebras in slices of $\Set$, and again the solution consists of (isoclasses 
  of) $P$-trees --- which notion needs refinement, though:
  corresponding to the fact
  that there are now many variables, indexed by the set $I$, there will now also
  be $I$-decorations on the edges.  A $P$-tree is now a tree whose edges are
  decorated by elements in $I$, and whose nodes are decorated by elements in
  $B$.  These decorations are subject to various compatibility constraints: if
  $b\in B$ decorates a node $x$, then the outgoing edge of $x$ must be decorated
  by $t(b)$; furthermore, as in the $1$-variable case, a bijection must be
  specified between the incoming edges of $x$ and the set $E_b$, and this
  bijection must be compatible with edge decorations in the sense that if an
  incoming edge corresponds to an element $e\in E_b$ then the decoration of that
  edge must be $s(e)$.
  
  While this may seem complicated at first look, there are plenty of natural
  examples of this, one of them being Feynman graphs, as we shall see in 
  Section~\ref{sec:trees-graphs}.  It is all greatly clarified by the following
  formalism.
\end{blanko}

\begin{blanko}{Trees.}\label{polytree-def}
  It was observed in \cite{Kock:0807} that operadic trees can be
  conveniently encoded by diagrams of the same shape as polynomial functors. 
  By definition, a
  {\em (finite rooted) tree} is a diagram of finite sets
\begin{equation}\label{tree}
\xymatrix{
    A & \ar[l]_s  M  \ar[r]^p & N  \ar[r]^t & A
}
\end{equation}
satisfying the following three conditions:
  
  (1) $t$ is injective
  
  (2) $s$ is injective with singleton complement (called the {\em 
  root} and denoted $1$).
  
  \noindent With $A=1+M$, 
  define the walk-to-the-root function
  $\sigma: A \to A$ by $1\mapsto 1$ and $e\mapsto t(p(e))$ for
  $e\in M$. 
  
  (3)  $\forall x\in A : \exists k\in \N : \sigma^{k}(x)=1$.
  
  The elements of $A$ are called {\em edges}.  The elements of $N$
  are called {\em nodes}.  For $b\in N$, the edge $t(b)$ is called
  the {\em output edge} of the node.  That $t$ is injective is just to
  say that each edge is the output edge of at most one node.  For
  $b\in N$, the elements of the fibre $M_b:= p^{-1}(b)$ are
  called {\em input edges} of $b$.  Hence the whole set
  $M=\sum_{b\in N} M_b$ can be thought of as the set of
  nodes-with-a-marked-input-edge, i.e.~pairs $(b,e)$ where $b$ is a
  node and $e$ is an input edge of $b$.  The map $s$ returns the
  marked edge.  Condition (2) says that every edge is the input edge
  of a unique node, except the root edge.
  Condition (3) says that if you walk towards the root, in a finite 
  number of steps you arrive there.
  The edges not in the image of $t$ are called {\em leaves}.
    
  The tree
  $$
  1 \leftarrow 0 \to 0 \to 1
  $$
  is the {\em trivial tree} \inlineDotlessTree.
\end{blanko}

\begin{blanko}{$P$-trees.}\label{Ptree}
  A great benefit of having trees and polynomials on the same footing is the
  efficiency in encoding and manipulating decorations of trees \cite{Kock:0807} (see also
  \cite{Kock:1109.5785,Kock:MFPS28,Kock:graphs-and-trees,Kock-Joyal-Batanin-Mascari:0706}).
  With a polynomial endofunctor $P$ fixed, given by a diagram
  $I \leftarrow E \to B \to I$,
  a {\em $P$-tree} is  by definition a diagram
  $$    \xymatrix{
    A \ar[d]_\alpha& \ar[l]  M\ar[d] \drpullback \ar[r] & N\ar[d] \ar[r] & 
    A\ar[d]^\alpha \\
    I  &\ar[l] E\ar[r] & B \ar[r]  &I ,
  }$$
  where the top row is a tree.  Hence nodes are decorated by elements in $B$,
  and edges are decorated by elements in $I$.  That the middle square is a
  pullback expresses that $n$-ary nodes of the tree have to be decorated by
  $n$-ary operations, and that a specific bijection is given.  When $P$ is a
  one-variable polynomial functor over $\Set$, this notion of $P$-tree
  specialises to the naive notion of $P$-trees from \ref{P-trees-naive}.
\end{blanko}

\begin{blanko}{Examples of $P$-trees.}
  Natural numbers are $P$-trees for the identity monad $P(X)=X$, and are also
  the set of operations of the list monad.  Planar trees are $P$-trees for $P$
  the list monad, and are also the set of operations of the
  free-non-symmetric-operad monad~\cite{Leinster:0305049}.  These two examples
  are the first entries of a canonical sequence of inductive data types
  underlying several approaches to higher category theory, the {\em opetopes}:
  opetopes in dimension $n$ are $P$-trees for $P$ a $\Set$-polynomial functor
  whose operations are $(n-1)$-opetopes~\cite{Kock-Joyal-Batanin-Mascari:0706};
  hence opetopes are higher-dimensional trees.

  Abstract trees cannot be realised as $P$-trees for any $\Set$-polynomial $P$.
  To realise  abstract naked trees as $P$-trees we need an endofunctor $P$ over 
  groupoids, as we now proceed to introduce.
\end{blanko}

\section{Groupoids}

\label{sec:groupoids}

While the theory of polynomial functors over sets is nice and useful for
dealing with trees, it comes short in
capturing important examples, like for example abstract (i.e.~non-planar) trees.
These ought to be $P$-trees for $P$ the terminal polynomial functor, which 
exists over groupoids but not over sets.
The case of particular interest in the present context
is that of
Feynman graphs: dues to their symmetries, these can be encoded as $P$-trees
over groupoids but not over sets, as explained in Section~\ref{sec:trees-graphs}.

It is thus necessary to upgrade the theory to groupoids.  With today's insight
into higher category theory, this upgrade is rather smooth: everything works
exactly in the same way as for sets, except that all notions have to be taken in
the homotopy sense: homotopy slices, homotopy pullbacks, homotopy adjoints,
homotopy quotients, homotopy cardinality, etc.  The reader unacquainted with
this machinery should not worry: it is legitimate to think of all the
constructions as taking place in the category of sets, but only bear in mind
that all problems with symmetries that would require special attention in the
category of sets, are {\em automatically} taken care of by the homotopy
formalism of groupoids.  We recall some basic facts about groupoids.
  
\begin{blanko}{Groupoids.}
  A {\em groupoid} is a category in which all arrows are invertible.  A {\em
  morphism} of groupoids is a functor, and we shall also need their natural
  transformations.  A useful intuition for the present purposes is that
  groupoids are `sets fattened with symmetries'.  From the correct homotopical
  viewpoint groupoids behave very much like sets.
  The homotopy viewpoint comes from the fact that groupoids are models for
  homotopy $1$-types, i.e.~topological spaces whose homotopy groups $\pi_k(X)$
  vanish for $k>1$. 

  A set can be considered a groupoid in which the only arrows are the identity arrows.
This defines a functor
$$
D: \Set \to \Grpd.
$$
Conversely, a groupoid $X$ gives rise to a set by taking its set of connected
components, i.e.~the set of isomorphism classes in $X$, denoted $\pi_0(X)$;
this defines a functor in the other direction (the left adjoint of $D$)
$$
\pi_0 : \Grpd \to \Set.
$$
Many sets arising in Combinatorics and Physics are actually $\pi_0$ of 
a groupoid, like when we say `the set of all trees' to mean the set of 
isomorphism classes of trees.

A group can be considered as a groupoid with only one object; the group 
elements being then the arrows. Conversely, for
each object $x$ in a groupoid $X$ there is associated a group, the {\em vertex
group}, denoted $\pi_1(x)$ or $\Aut(x)$, which consists of all the arrows from
$x$ to itself.

The homotopy notations $\pi_0$ and $\pi_1$ match their counterparts in Topology.
\end{blanko}

\begin{blanko}{Equivalences of groupoids; discreteness and contractibility.}
  An {\em equivalence} of groupoids is just an equivalence of categories, i.e.~a
  functor admitting a pseudo-inverse.  Pseudo-inverse means that the two
  composites are not necessarily exactly the identity functors, but are only
  required to be isomorphic to the identity functors.  A morphism of groupoids
  is an equivalence if and only if it induces a bijection on $\pi_0$, and an
  isomorphism at the level of $\pi_1$.  This is the analogue of a homotopy
  equivalence in Topology.

  A groupoid $X$ is called {\em discrete} if it is equivalent to a set considered as 
  a groupoid; this set can then be taken to be $\pi_0(X)$.
  Another way of saying the same is that all vertex groups are
  trivial: $\pi_1(x)=1$ for all objects $x\in X$, so all the information is 
  stored in $\pi_0$.
  (There is a potential
  risk of confusion with the word `discrete': in settings where one considers
  Lie groupoids (as in \cite{Connes:1994}), the word `discrete' usually designates 
  groupoids whose underlying topological space is discrete.)

  A groupoid is called {\em contractible} if it is equivalent to a singleton set.

  We are interested in groupoids up to equivalence, and for
  this reason many familiar $1$-categorical notions, such as pullback and fibre,
  are not appropriate, as they are not invariant under equivalence.  The good
  notions are the corresponding {\em homotopy} notions, which we briefly recall.
\end{blanko}

\begin{blanko}{Pullbacks and fibres.}
  Given a diagram of groupoids $X,Y,S$ indicated by the solid arrows,
  $$
  \xymatrix{
  X\times_S Y \drpullback \ar@{-->}[r]\ar@{-->}[d]
  &Y\ar[d]^-{g}\\X\ar[r]_-{f}&S
  }
  $$
  the {\em homotopy pullback} is
  the groupoid $X\times_S Y$ whose objects 
  are triples
  $(x,y,\phi)$
  with $x \in X$, $y\in Y$ and  $\phi:fx\to gy$ an arrow of $S$, and whose
  arrows are pairs $(\alpha,\beta):(x,y,\phi)\to(x',y',\phi')$ consisting
  of
  $\alpha: x \to x'$ an arrow in $X$ and $\beta: y \to y'$ an arrow in $Y$
  such that the following diagram commutes in $S$:
  $$
  \xymatrix{
  fx\ar[r]^-\phi\ar[d]_-{f(\alpha)}&gy\ar[d]^-{g(\beta)}\\
  fx'\ar[r]_-{\phi'}&gy'.
  }
  $$
  The homotopy pullback is an example of a homotopy limit, and as such enjoys a
  universal property analogous to that of ordinary pullbacks of sets.
  From now on, {\em pullback} means homotopy pullback, indicated in the square
  with an angle symbol.

  The \emph{homotopy fibre} $E_b$ of a morphism $p:E\to B$ over an object $b$ in $B$ is 
  the homotopy pullback of $p$ along the inclusion map
  $\xymatrix{1\ar[r]^-{\name b}&B}$ that sends the unique point in the 
  trivial groupoid to $b$:
  $$
  \xymatrix{
  E_b\drpullback\ar[r]\ar[d]
  &E\ar[d]^-{p}\\1\ar[r]_-{\name b}&B.
  }
  $$
  (Note that the homotopy fibre $E_b$ is not in general a
  subgroupoid of $E$, although the map $E_b \to E$ is always faithful.)
\end{blanko}

\footnotesize

\begin{blanko}{Homotopy pullbacks are common in Quantum Field Theory.}
  Say we want to substitute a graph $H\in \kat{Gr}$ (graphs) into the marked vertex $v$ of 
  another graph $G\in \kat{Gr}\upperstar$ (graphs with a marked vertex).
  For this, it is clearly necessary that the set of external lines 
  of $H$ (the residue of $H$) matches the set of lines incident to the vertex 
  $v\in \kat{Star}$ (connected graphs without internal lines).  Obviously the
  two sets are not exactly the same, so what is really meant is that there is
  a bijection between the set of external lines of $H$ and the set of lines 
  incident to
  $v$, and it is also necessary to know {\em which} bijection is used.  
  Altogether, we are precisely considering an element of the homotopy pullback
  $$\xymatrix{
      \drpullback\ar[r]\ar[d] & \kat{Gr}\upperstar \ar[d]^{\text{return the marked vertex}} \\
     \kat{Gr} \ar[r]_{\res} & \kat{Star}
  }$$
\end{blanko}

\normalsize

\begin{blanko}{Homotopy quotients.}
  Whenever a group $G$ acts on a set or a groupoid $X$, the
  {\em homotopy quotient} $X/G$ is the groupoid obtained by gluing in a path
  (i.e.~an arrow)
  between $x$ and $y$ for each $g\in G$ such that $xg=y$.
  It is a special case of a homotopy colimit. 
  (The notation $X/\!/G$ is often used~\cite{Baez-Dolan:finset-feynman}.) 
  If $G$ acts on the point groupoid $1$, then 
  $1/G$ is the groupoid with one object and vertex group $G$.

  If $p:E\to B$ is a morphism of groupoids, for $b\in B$ the
  `inclusion' of the homotopy fibre $E_b \to E$ is faithful but not full in general
  (see \cite{Leinster:BCT} for terminology).  But
  $\Aut(b)$ acts on $E_b$ canonically,
  \begin{eqnarray*}
    E_b \times \Aut(b) & \longrightarrow & E_b  \\
    ((x,\phi), \beta) & \longmapsto & (x,\beta\circ \phi) ,
  \end{eqnarray*}
  and the homotopy quotient $ E_b /\Aut(b) $
  provides exactly the missing arrows, so as to make the natural map $E_b
  /\Aut(b) \to E$ fully faithful.  Since every object $x\in E$ must map to some
  connected component of $B$, we find the equivalence
\begin{equation}\label{hosum}
E \ \simeq
\sum_{b\in \pi_0 B} E_b/\Aut(b) \ =: \ \int^{b\in B} E_b,
\end{equation}
expressing $E$ as the {\em homotopy sum} of the fibres.
  While an ordinary sum (disjoint union) is a colimit indexed by a set, a
  homotopy sum is a (homotopy) colimit indexed over a groupoid.
  
Our main use of such `fibrations' $p:E\to B$ 
is to deliver a family of groupoids, namely 
the fibres $E_b$.  To reconstruct the fibration from the family of groupoids, it is
necessary to know how they are glued together, which amounts to having actions
of the vertex groups of the base on the fibres.  We will often let these actions
be implicit, and specify a family as a collection of groupoids, like this:
$$
(E_b \mid b\in B).
$$
\end{blanko}

\footnotesize

\begin{blanko}{Homotopy sums are common in Quantum Field Theory.}
  The characterisation \eqref{hosum} shows that the homotopy sum can be 
  calculated as sum over the set of isoclasses, with each summand
  being quotiented by the natural action of the 
  automorphism group:
  $$
  \int^{b\in B} E_b = \sum_{b\in \pi_0 B} \frac{E_b}{\Aut(b)}
  $$
  Summing over isoclasses of objects and dividing out by symmetry factors is
  common in Physics and Combinatorics.  Virtually always this can be explained
  as a homotopy sum of groupoids.  See
  \cite{GalvezCarrillo-Kock-Tonks:1207.6404} for further exploitation of this
  viewpoint.
\end{blanko}

\normalsize

\section{Polynomial functors over groupoids, and $P$-trees}

\label{sec:P}

\begin{blanko}{Slices.}
  If $I$ is a groupoid, the {\em homotopy slice} $\Grpd_{/I}$ is the
  category  whose 
  objects are maps $X \to I$ and whose arrows are
  triangles with a $2$-cell
  $$
  \xymatrix@C=3ex{
  X \ar[rr] \ar[rd] & \ar@{}[d]|{\Rightarrow} & Y \ar[ld] \\
  &I .&
  }$$
  (In reality the groupoid slice should be construed as a $2$-category, in order
  for the following universal constructions to be correct.  In the present
  exposition we gloss over these subtleties.)
\end{blanko}

\begin{blanko}{Basic adjoints between slices.}\label{adj}
  Taking homotopy pullback along a morphism of groupoids $f:B\to A$ defines a
  functor between the slice categories
  $$
  f\upperstar 
  :\Grpd_{/A}\to \Grpd_{/B} ,
  $$
  which in family notation  is given by the formula
  $$
  f\upperstar(X_a\mid a\in A) \ = \ (X_{f(b)} \mid b\in B) ,
  $$
  completely analogous to the $\Set$ case (Formula~\eqref{pbk}).
  
This has a homotopy left adjoint
$$f\lowershriek  :\Grpd_{/B}\to \Grpd_{/A}$$ 
given in family notation by
$$
f\lowershriek(Y_b \mid b\in B) \ = \ (\textstyle{\int^{b\in B_a} Y_b} \mid a\in A) \,,
$$
just as Formula~\eqref{sum} in the $\Set$ case.

The pullback functor also has a homotopy right adjoint 
$$f\lowerstar  :\Grpd_{/B}\to \Grpd_{/A}.$$ 
  For $Y \to B$, the
  fibre of $f\lowerstar Y$ over $a\in A$ can be described explicitly as the
  mapping groupoid
  \vspace*{-4pt}
  $$
  (f\lowerstar Y)_a = \Map_{/B}( B_a , Y) .
  $$

(The homotopy adjoint properties are expressed by 
natural equivalences of mapping spaces
$\Map_{/A}(f\lowershriek Y,X)\simeq
\Map_{/B}(Y,f\upperstar X)$ and $\Map_{/B}(f\upperstar X,Y)\simeq
\Map_{/A}(X,f\lowerstar Y)$, but this will not actually be needed.)
\end{blanko}

\begin{blanko}{Sets versus groupoids.}
  It should be observed that all the fancy homotopy notions for groupoids
  actually specialise to the classical notions when the groupoid is discrete.
  For example, in the homotopy pullback, if $S$ is discrete, then there are 
  no non-trivial $\phi$, so the homotopy pullback reduces to the ordinary 
  pullback.  Similarly, the homotopy sum indexed by a discrete groupoid is
  precisely the ordinary sum (over a strictly discrete skeleton).
\end{blanko}

\begin{blanko}{Polynomial functors.}
  A {\em polynomial} is a diagram of groupoids
  $$
I \stackrel s \longleftarrow E \stackrel p \longrightarrow B \stackrel t 
\longrightarrow J .
$$
The associated {\em polynomial functor}
is given as the composite
$$
\Grpd_{/I} \overset{s\upperstar}\longrightarrow
\Grpd_{/E} \overset{p\lowerstar}\longrightarrow
\Grpd_{/B} \overset{t\lowershriek}\longrightarrow
\Grpd_{/J}  .
$$
 (Here of course we are talking about homotopy slices, and upperstar, lowerstar 
 and lowershriek refer to the adjunctions in \ref{adj}.)
  The intuition is that $B$ is a collection of operations, the
  arity of an operation $b\in B$ is (the size of) the fibre $E_b$, and that each
  operation is typed: the output type of $b$ is $t(b)$, and the input types are
  the $s(e)$ for $e\in E_b$.
  
  We shall be concerned only with the case where $p$ has finite discrete fibres.
  In this case, there is the following explicit formula for the polynomial 
  functor:
  $$
  (X_i \mid i\in I) \longmapsto (  \int^{b\in B_j} \prod_{e\in E_b} x_{se} \mid j\in J)  .
  $$
\end{blanko}

\begin{blanko}{Morphisms of trees (cf.~\cite{Kock:0807}).}\label{sub}
  A {\em tree embedding} is by definition a diagram (of sets)
  \begin{equation}
    \label{equ:cartmorphism}
  \xymatrix{
  A' \ar[d]_\alpha& \ar[l]  M'\ar[d] \drpullback \ar[r] & N'\ar[d] \ar[r] & A'\ar[d]^\alpha  \\
  A  &\ar[l] M\ar[r] & N \ar[r]  &A ,
  }
  \end{equation}
  where the rows are trees.
  The terminology is justified by the fact that each of
  the components of such a map is necessarily injective; this follows from the 
  tree axioms \cite{Kock:0807}.  Hence the category of
  trees and tree embeddings, denoted $\TEmb$, is mostly concerned with subtrees,
  but note that it also contains automorphisms of trees.
  The fact that the middle square is a pullback amounts to having,
  for each node $b$ of the first tree, a bijection between
  the incoming edges of $b$ and the incoming edges of the image of $b$.
  In other words, a tree embedding is arity preserving.

  A tree embedding is {\em root preserving} when it sends the root to the root.
  In formal terms, these are diagrams \eqref{equ:cartmorphism}
  such that also the left-hand square is cartesian \cite{Kock:0807}.
  
  \label{ideal}
  An {\em ideal embedding} (or an {\em ideal subtree}) is a subtree $S$
  in which for every edge $e$, all the descendant edges and nodes of $e$ are also in $S$.
  There is one ideal subtree $D_x$
  generated by each edge $x$ in the tree.  The ideal embeddings are characterised
  as having also the right-hand square of \eqref{equ:cartmorphism} cartesian \cite{Kock:0807}.
  
  Ideal embeddings and root-preserving embeddings admit pushouts along each 
  other in the category $\TEmb$ \cite{Kock:0807}.
  The most interesting case is pushout over a 
  trivial tree: this trivial tree 
  is then the root of one tree and a leaf of another tree,
  and the pushout is the grafting onto that leaf.
\end{blanko}

\begin{blanko}{$P$-trees.}\label{PtreeGr}
  For $P$ a fixed finitary polynomial endofunctor defined by groupoids
  $$
  I \stackrel s\leftarrow E \stackrel p \to B \stackrel t \to I,
  $$
  as in the set case, a {\em $P$-tree} is a diagram
$$    \xymatrix{
  A \ar[d]_\alpha& \ar[l]  M\ar[d] \drpullback \ar[r] & N\ar[d] \ar[r] 
& A\ar[d]^\alpha \\
  I  &\ar[l] E\ar[r] & B \ar[r]  &I ,
  }$$
  where the top row is a tree, but now 
  the squares
  are commutative only up to isomorphism, and it is important that the isos 
  be specified as
  part of the structure.
  Unfolding the definition, we see that a $P$-tree is a
  tree whose edges are decorated in $I$, whose nodes are decorated
  in $B$, and with the additional structure of a bijection for each
  node $n \in N$ (with decoration $b \in B$) between the set of
  input edges of $n$ and the fibre $E_b$, subject to the
  compatibility condition that such an edge $e\in E_b$ has
  decoration $s(e)$, and the output edge of $n$ has decoration isomorphic
  to $t(b)$.
  
  (Note that the natural thing to say is `equivalence' instead of `bijection', but 
  under the assumption that the fibres of $p$ are discrete, we can say 
  bijection instead of equivalence, provided we understand that the discrete
  groupoid is strictly discrete (i.e. replace it by its $\pi_0$).)
  
  In Section~\ref{sec:trees-graphs} we shall see how to encode Feynman
  graphs as $P$-trees.  To properly account for symmetries of graphs,
  it is essential that the 
  representing diagram $I \leftarrow E \to B \to I$ be of groupoids, 
  not just sets.
\end{blanko}

\begin{blanko}{Remark.}
  Both the definition of tree embedding and the definition of $P$-tree
  are special cases of the notion of cartesian morphism of polynomial 
  endofunctors, as we shall see in \ref{cart}.
\end{blanko}

\begin{blanko}{Trivial trees.}
  A {\em trivial $P$-tree} is a $P$-tree without nodes, so it is just a single edge
  decorated by an element in $I$.  Denote by $\kat{Triv}$ the groupoid of 
  trivial trees, and denote by $\kat{Triv}\comma P$ the groupoid of trivial $P$-trees.
\end{blanko}

\begin{lemma}
  We have
  $$
  \kat{Triv}\comma P \simeq I .
  $$  
\end{lemma}

\begin{cor}
  The automorphism group of a trivial $P$-tree decorated by $i\in I$ is
  $\Aut_I(i)$.
\end{cor}

\begin{blanko}{Corollas.}
  A {\em corolla} is a tree with exactly one node.  We denote by $\kat{Cor}$ the
  groupoid of all corollas.  A {\em $P$-corolla} is a $P$-tree whose underlying
  tree is a corolla.  The groupoid of $P$-corollas is denoted $\kat{Cor}\comma P$.
\end{blanko}

\begin{lemma}\label{lem:Cor-P}
  We have
  $$
  \kat{Cor}\comma P \simeq B .
  $$
\end{lemma}
It follows from Lemma~\ref{lem:Cor-P} that the automorphism group of
a $P$-corolla labelled by $b\in B$ is $\Aut_B(b)$.

\begin{blanko}{Combinatorial interpretation of functor evaluation.}
  \label{Peval}
  In view of the previous lemma, we can picture the elements in
  $B$ as $P$-corollas:
  \begin{equation}
    \begin{texdraw}
  \linewd 0.8 \footnotesize
  \move (0 0)
  \bsegment
    \move (0 0) \lvec (0 12) \onedot 
    \lvec (-14 25)
    \move (0 12) \lvec (-5 28) 
    \move (0 12) \lvec (5 28) 
    \move (0 12) \lvec (14 25)
    \move (8 10) \htext {$b$}
    \move (0 42) \htext {$E_b$}
    \move (0 32) \htext {$\overbrace{\phantom{xxxxx}}$}
  \esegment
  \end{texdraw}\end{equation}
  With this interpretation, we can trace through the definition of the 
  polynomial functor
  $P$
  to see that for $X \to I$ an $I$-indexed groupoid, the groupoid $P(X)$
  is described as that of $X$-decorated $P$-corollas. 
  Precisely, $P(X)$ has as
  objects pairs $(b,f)$ where $b\in B$ and $f: E_b \to X$ is an $I$-map,
  so actually a $2$-cell
    $$
  \xymatrix@C=3ex{
  E_b \ar[rr]^f \ar[rd] & \ar@{}[d]|{\Rightarrow} & X \ar[ld] \\
  &I .&
  }$$
  An arrow in $P(X)$ from $(b,f)$ to $(b',f')$ consists of an arrow
  $\beta:b\isopil b'$ in $B$ together with a $2$-cell (over $I$)
      $$
  \xymatrix@C=3ex{
  E_b \ar[rr]_\sim^{\beta\lowershriek} \ar[rd]_f & \ar@{}[d]|{\Rightarrow} & E_{b'} \ar[ld]^{f'} \\
  &X&
  }$$
  where $\beta\lowershriek : E_b \isopil E_{b'}$ is the map induced by $\beta$:
  \begin{eqnarray*}
   \beta\lowershriek : E_b & \longrightarrow & E_{b'}  \\
    (e,\phi) & \longmapsto & (e,\beta\circ \phi)
  \end{eqnarray*}
  
  Note that $P(X)$ is naturally a groupoid over $B$ and hence also a groupoid
  over $I$.  As a groupoid over $B$, we see that the description is given by 
  describing the fibres, and saying how the groups $\Aut(b)$ acts on the fibre 
  over $b$.  This is the description of $P(X)$ as the homotopy sum of its 
  fibres.
  
  We should also say what $P$ does on a morphism $a:X \to Y$ in $\Grpd_{/I}$:
  its value on $a$  sends an object $(b,f)$ to the object $(b,a\circ f)$.
\end{blanko}

\begin{prop}
  To give an automorphism of a $P$-tree $T$, assumed to have a bottom node $R$,
  is to give an automorphism $\sigma: R \isopil R$ of the corolla $R$ together
  with an isomorphism $D_x \isopil D_{\sigma x}$ for each $x$ incoming leaf of 
  the corolla $R$:
  $$
  \Aut(T) \simeq \Aut(R) \times \prod_{x\in E_b} 
  \operatorname{Iso}(D_x, D_{\sigma x}) . 
  $$
\end{prop}

\section{Abstract combinatorial Dyson--Schwinger equations}

\label{sec:groupoid-DSE}

\begin{blanko}{Fixpoint equations, general case.}
  We keep fixed a polynomial endofunctor $P: \Grpd_{/I} \to \Grpd_{/I}$ given by
  groupoids
  $$
  I \stackrel s \leftarrow E \stackrel p \to B \stackrel t \to I
  $$
  and assume that $p$ has finite discrete fibres.
  We are interested in solving the fixpoint equation
  $$
  X \isopilback 1 + P(X) .
  $$
  Here $X$ is an object in the slice category $\Grpd_{/I}$, that is, an
  $I$-indexed family of groupoids, and $1$ denotes the terminal object in this 
  slice category, which is the family $\id:I \to I$, whose members are all 
  terminal.
  
  {\em Solving} this equation means finding the least fixpoint, and more 
  formally,
  finding a homotopy initial $(1+P)$-algebra (where $1$ is the constant functor 
  $1$, and the sum is sum of functors).  
  Some subtleties deserve explanation here: in the category of
  $(1+P)$-algebras, the morphisms are not just commutative squares like in
  \eqref{eq:P-alg map}: the square is allowed to commute only up to a specified
  invertible natural transformation.  Homotopy initial means that we are not 
  dealing
  with an isomorphism, but only with an equivalence of groupoids, and 
  accordingly, homotopy initial objects are unique only up to unique equivalence.
  So when we say below that the solution is the groupoid of $P$-trees, we
  are entitled to substitute this groupoid by any equivalent groupoid, like
  for example taking a skeleton of it.  
\end{blanko}

\begin{theorem}\label{thm:hoinit-gr}
  The (homotopy) initial $(1+P)$-algebra is the groupoid of $P$-trees.
\end{theorem}

We prove this theorem by hand, establishing that the groupoid of $P$-trees
actually solves the equation.  A more formal proof is possible, valid even
for $\infty$-groupoids
\cite{Gepner-Kock}.  
\begin{proof}
  Denote by $W$ the groupoid of $P$-trees, considered as a groupoid over $I$,
  that is an object in $\Grpd_{/I}$.  It is the disjoint union of
  $\kat{Triv}\comma P$ and the groupoid $W^+$ of $P$-trees with a bottom node.
  The first is the terminal object $I\to I$ in $\Grpd_{/I}$, so we are done if
  we can exhibit a natural equivalence of groupoids (over $I$)
  $$
  W^+ \isopilback P(W) .
  $$
  The map is described as follows: according to \ref{Peval}, an object in $P(W)$
  is a pair $(b,f)$ where $b\in B$ and $f: E_b \to W$ over $I$.  We can
  interpret $b$ as a $P$-corolla, and since its leaves $x$ match the roots of
  the trees $f(x)$, we can glue all these trees onto $b$ to obtain a tree (which
  clearly has a bottom node, namely $b$).  Conversely, given a $P$-tree with a
  bottom node $b$, consider for each $x\in E_b$ the ideal tree $D_x$, as in 
  \ref{ideal}.  This data
  is precisely the data of $(b,f)$.  So on objects we can go back and forth
  easily.  To make this correspondence an equivalence of groupoids, it remains
  to check that the automorphism groups match up.  We have already computed the
  automorphism group of a tree $T$ with a node: to give an automorphism of $T$
  is to give an automorphism group of the bottom node $R$ and for each $x$ an
  incoming edge of $R$, an isomorphism $D_x \isopil D_{\sigma x}$.  Now we also know that
  the automorphism group of $R$ is $\Aut(b)$, so now it is easy to see that
  altogether this is precisely the automorphism group of the data $(b,f)$:
  indeed, that is an automorphism of $b$ together with an automorphism of $f:
  E_b \to W$.
  
  We omit the proof that no smaller groupoid solves the equation, which amounts 
  to initiality.
  This argument (similar to the discrete case~\cite{Kock:0807})
  goes by induction on the number of nodes.
\end{proof}

\section{Bialgebras and $B_+$-operators}

\label{sec:groupoid-Green}

Recall that subtree means a morphism of trees in the sense of \ref{sub}
(this is not required to preserve the 
root).  Hence subtrees are always full, in the sense that if a node
is contained in the subtree then so are all the incident edges to that node.

\begin{blanko}{The bialgebra of $P$-trees (cf.~\cite{Kock:1109.5785}).}
  \label{treebialg-general}
  We deal with $P$-trees here, but for short we just say tree.
  A {\em cut} of a tree is defined to be a subtree containing the
  root.  If $c:S\subset T$ is a subtree containing the root, then each leaf $e$
  of $S$ determines an ideal subtree of $T$, namely the tree $D_e$
  consisting of $e$ (which
  becomes the new root) and all the edges and nodes above it (\ref{ideal}).  This is still
  true when $e$ is also a leaf of $T$: in this case, the ideal tree is the
  trivial tree consisting solely of $e$.  Note also that the root
  edge is a subtree; the ideal tree of the root edge is of course the tree
  itself.  This is the analogue of the cut-below-the-root in the combinatorial
  case (\ref{CK}).  For a cut $c:S\subset T$, define $P_c$ to be the forest consisting
  of all the ideal trees generated by the leaves of $S$.
  
  Let $\BB=\BB_P$ be the free commutative algebra $\Q[T \in \pi_0 \tree]$
  on the set of isomorphism 
  classes of $P$-trees.  It becomes a bialgebra with
  comultiplication defined on the generators by
  \begin{eqnarray*}
  \Delta:  \BB & \longrightarrow & \ \ \BB \otimes \BB  \\
    T & \longmapsto & \sum_{c:S\subset T}\!\! P_c \otimes S ,
  \end{eqnarray*}
  as exemplified by 

  \begin{center}\begin{texdraw}
    \htext (0 0) {$\Delta($}
   \rmove (12 0) \twotree
     \rmove (15 0)
    \htext  {$) \ \ = $}
        \rmove (23 0) \zerotree
    \rmove (3.5 0) \zerotree
    \rmove (3.5 0) \zerotree

          \rmove (11 0) \htext {$\tensor$}
   \rmove (14 0) \twotree

	            \rmove (18 0) \htext {$+$}

    \rmove (20 0) \onetree  \rmove (5 0) \zerotree 
    
    \rmove (12 0) \htext {$\tensor$}
   \rmove (12 0) \onetree
   
      \rmove (19 0) \htext {$+$}

    \rmove (21 0) \twotree
    
    \rmove (13 0) \htext {$\tensor$}
   \rmove (10 0) \zerotree
  \end{texdraw}\end{center}  
  The description of cuts in terms of subtrees and ideal subtrees means that, for
  operadic trees, cuts are not allowed to go above the leaves, and that cutting
  an edge does not delete it, but really cuts it.  This is an essential point, as 
  deleting edges would break the given arity
  bijections in the decorations.

  Coassociativity of the comultiplication
  can be checked by hand, but it is not necessary: it is
  a consequence of general principles \cite{Galvez-Kock-Tonks:1512.07573}:
  any polynomial endofunctor defines a monoidal decomposition space, and any
  monoidal decomposition space (subject to some finiteness 
  conditions~\cite{Galvez-Kock-Tonks:1512.07577}) defines a bialgebra.
\end{blanko}
  
\begin{blanko}{Core.}\label{core-bialghomo}
  There is a canonical bialgebra homomorphism from any such bialgebra of
  $P$-trees to the Connes--Kreimer bialgebra $\CK$, given by taking
  core~\cite{Kock:1109.5785}: this amounts to forgetting the $P$-decorations and
  shaving off all leaves and roots.  In other words, the core of a $P$-tree is
  the combinatorial tree given by its inner edges.
\end{blanko}

\begin{blanko}{Gradings.}
  The bialgebra $\BB_P$ has two natural gradings: an $\N$-grading by the number
  of nodes, and a $\Z$-grading by `number of leaves minus number of roots'.
  (Neither of these gradings has actually anything particular to do with the
  polynomial endofunctor $P$; they refer only to the underlying tree.  On the
  other hand, both gradings admit refinements and variations depending on $P$:
  for example, certain nodes and edges might have different weight.)

  The node grading is an instance of {\em length filtration}, a 
  general notion for decomposition spaces~\cite{Galvez-Kock-Tonks:1512.07577},
  which specialises to the familiar notion of length in locally finite posets.
  
  The second grading is called the {\em operadic grading}, and will play an important 
  role in what follows.  Note that for the operadic grading, nullary trees have
  degree $-1$, and that by taking disjoint union with nullary trees, forests
  of arbitrary negative degree can be obtained.  Particular examples of $P$
  may of course prohibit nullary nodes altogether, and hence turn the operadic
  grading into an $\N$-grading too.
\end{blanko}

\begin{blanko}{Remark in defence of bialgebras that are not quite connected.}
  \label{bla:nonconnected}
  The bialgebra $\BB_P$ is not connected for the length grading: degree zero is
  spanned by all the nodeless forests.  (It is not connected for the operadic
  grading either.)  In particular, it is not a Hopf algebra, as used in the
  Connes--Kreimer approach to renormalisation.  However, it can be
  shown~\cite{Kock:1411.3098} that Hopf algebra renormalisation works for
  bialgebras of $P$-trees too, the essential point being that degree zero is
  spanned by group-like elements.  For the same reason, the connected quotient
  of $\BB_P$ exists, obtained by equating all nodeless forests to $1$.
  However, this process would clearly destroy the operadic grading, and with it 
  algebraic structure that will be exploited in an essential way in what 
  follows.
\end{blanko}

\begin{blanko}{The bialgebra at the objective level.}
  Although for applications one may be interested primarily in the
  bialgebra $\BB_P$ at the level of numbers, vector spaces and linear maps, it
  is an important aspect of the theory, in the larger picture, that all the
  constructions can be lifted to the objective level of objects, slices, and
  functors, cf.~\cite{GalvezCarrillo-Kock-Tonks:1207.6404}.
  We briefly discuss this, as it will be
  useful in order to describe both $B_+$-operators and Green functions.

  In a nutshell, Algebraic Combinatorics is
  concerned with vector spaces $\Q_{\pi_0 S}$ spanned by iso-classes of
  certain combinatorial objects forming a groupoid $S$.
  Further algebraic structure is induced by the combinatorics of these objects,
  and often such operations can be realised as (homotopy) cardinality of
  operations at the combinatorial level, also referred to as the objective 
  level \cite{Galvez-Kock-Tonks:1602.05082}.
  The objectification of the vector space $\Q_{\pi_0 S}$
  is the slice category $\Grpd_{/S}$ of groupoids 
  over $S$.  For $X \to S$ an object in here, and $s\in S$, the homotopy fibre $X_s$ is
  the objective version of the coefficient of $s$ in a linear combination.
  Linear maps $\Q_{\pi_0 T} \to \Q_{\pi_0 S}$ are objectified by linear functors
  $\Grpd_{/T} \to \Grpd_{/S}$, which in turn are given by
  {\em spans} $T \leftarrow M \to S$ (see \cite{Galvez-Kock-Tonks:1602.05082} for terminology);
  this is a groupoid indexed by the product 
  $T\times S$,
  just as linear maps are given by matrices. 
  There is a symmetric monoidal structure on the category of groupoid slices 
  given by
  $$
  \Grpd_{/I} \tensor \Grpd_{/J} = \Grpd_{/I \times J}
  $$
  objectifying the familiar isomorphism $\Q_I \tensor \Q_J = \Q_{I \times J}$.
  The neutral object is $\Grpd=\Grpd_{/1}$, the slice over the terminal object,
  corresponding to the ground field $\Q$.  In all these cases, the 
  linear-algebra notion is obtained by taking (homotopy) cardinality of the 
  groupoid notions~\cite{Galvez-Kock-Tonks:1602.05082}.

  We let $\tree$ denote the groupoid of
  $P$-trees. Let
  $\forest = \forest_1$ denote the groupoid of $P$-forests (i.e.~disjoint unions of 
  $P$-trees).  Denote by $\forest_2$ the groupoid of forests with a cut,
  and denote by $\forest_0$ the groupoid of nodeless forests, i.e.~disjoint unions
  of trivial trees.
  The reason for the indices
  is that these groupoids form the beginning of a simplicial 
  object~\cite{Galvez-Kock-Tonks:1512.07573}, \cite{Kock-Weber:1609.03276}
  $$\xymatrix{
  \forest_0  & \ar@<-2pt>[l]\ar@<2pt>[l]
  \forest_1  & \ar@<-4pt>[l]  \ar[l]  \ar@<4pt>[l]
  \forest_2  & \ar@<-6pt>[l]\ar@<-2pt>[l]\ar@<2pt>[l]\ar@<6pt>[l]
  \forest_3 &
  \cdots
  }$$
  where in general $\forest_k$ is the groupoid of forests with $k-1$ compatible 
  cuts.  The face maps consist in forgetting cuts (e.g.~$d_1: \forest_2 \to 
  \forest_1$) or removing the crown (e.g.~$d_0: \forest_2 \to \forest_1$) or the 
  bottom forest (e.g.~$d_2: \forest_2 \to \forest_1$).  It is a fundamental fact 
  (\cite[Key Lemma]{GalvezCarrillo-Kock-Tonks:1207.6404}, see also 
  \cite{Kock-Weber:1609.03276}), that we have
  $$
  \forest_2 \simeq \forest_1 \times_{\forest_0} \forest_1
  $$
  which is to say that this simplicial object is a category object (i.e.~the 
  fat nerve of a category as in \cite{Galvez-Kock-Tonks:1512.07573}).
  This important property is not shared by the simplicial object of 
  combinatorial trees, as explained in detail in \cite{Galvez-Kock-Tonks:1612.09225}.
  
  The objectification of 
  $\Q_{\pi_0 \forest_1}$ is hence the groupoid slice $\Grpd_{/\forest_1}$.
  The comultiplication in the bialgebra of $P$-trees is now objectified by the
  linear functor
  $$
  \Grpd_{/\forest_1} \stackrel{d_1\upperstar}\longrightarrow \Grpd_{/\forest_2} 
  \stackrel{(d_2,d_0)\lowershriek}\longrightarrow 
  \Grpd_{/\forest_1\times\forest_1}
  $$
  given by the span
  $$
  \forest_1 \stackrel{d_1}\longleftarrow \forest_2 \stackrel{(d_2,d_0)}\longrightarrow 
  \forest_1\times\forest_1 .
  $$
  Here the leftward arrow sends a forest with a cut to the total forest
  (just forget the cut), whereas the rightward arrow sends a forest with a
  cut to the
  pair consisting of the top and bottom forests.
\end{blanko}

\begin{blanko}{$B_+$-operators.}\label{B+}
  The bialgebra of $P$-trees has canonical $B_+$-operators, which are
  essentially the components of the structure equivalence
  $$
  \tree \isopilback 1 + P(\tree).
  $$
  Indeed, the formula for $P$ gives
  $$
  P(X) = \int^{b\in B} \Map_I(E_b,X) ,
  $$
  and the map from
  $P(\tree)$ to $\tree$ is therefore the homotopy sum of maps (one for each $b\in B$)
  \begin{equation}\label{preB}
    \Map_I(E_b,\tree) \to \tree .
  \end{equation}
  The very construction of the solution $\tree$ showed
  that this map takes an $E_b$-indexed family of trees and returns the tree
  obtained by grafting all these trees onto the corolla $b$.  This is essentially
  the $B_+$-operator corresponding to $b\in B$. 
  
  We can analyse this a bit further, by noting the
  canonical equivalence
  $$
  \Map_I(E_b,\tree) \simeq \forest_1 \times_{\forest_0} \{b\} ,
  $$
  which states precisely that an $E_b$-indexed family of trees is the same thing
  as a forest equipped with a bijection from its set of roots to the set of 
  leaves of the corolla $b$. 
  This in turn can be interpreted in terms of the canonical map
  $$
  \forest_1 \times_{\forest_0} \forest_1 \simeq \forest_2 \stackrel{d_1}\longrightarrow 
  \forest_1 .
  $$
  This map takes two forests and glues the first onto the leaves of the second
  (the fibre product expresses the condition that the roots of the first forest 
  match the leaves of the second forest).  If we demand the bottom forest to 
  be a tree, clearly the result will be a tree again, and hence the map 
  restricts to 
  $
     \forest_1 \times_{\forest_0} \tree 
     \longrightarrow
      \tree $,
  and we can restrict further to require the bottom tree to be a corolla, or
  indeed to be the specific corolla $b$:
  $$\xymatrix{
     \forest_1 \times_{\forest_0} \{b\} \ar[rr]\ar[d] && \{b\} \ar[d] \\
     \forest_1 \times_{\forest_0} B \ar[rr]\ar[d] && B \ar[d] \\
     \forest_1 \times_{\forest_0} \tree \ar[rr]\ar[d] && \tree \ar[d] \\
     \forest_1 \times_{\forest_0} \forest_1 \ar[rr] && \forest_1
  }$$
  Now the top morphism is precisely 
  the preliminary $B_+$-operator given in 
  \eqref{preB}.

  It remains to observe that the final version of the $B_+$-operator should take
  as input any forest, not specifically an $E_b$-indexed family of trees.  The
  final version of the $B_+$-operator corresponding to $b$ is the linear map
  given by the span
  $$
  \forest_1 \longleftarrow \forest_1 \times_{\forest_0} \{b\} \longrightarrow 
  \tree   .
  $$

  These are the individual $B_+$-operators, one for each $b\in B$.
  There is also a global $B_+$-operator, obtained as the homotopy sum of
  all the individual ones.  This comes about very formally: it is simply the
  diagram
  $$
  \forest_1 \longleftarrow 
   \forest_1 \times_{\forest_0} B  \longrightarrow \tree
   $$
   giving the global $B_+$-operator 
$$
\Grpd_{/\forest_1} \longrightarrow \Grpd_{/\tree}
$$
by pullback along the projection, and then lowershriek
over to $\tree$.

Note that naturally, without choices or tricks, the global $B_+$-operator is a
homotopy sum of the individual $B_+$-operators, which can be spelled out as a
sum weighted by the symmetry factors of each $b$.  (The appearance of these
symmetry factors is a familiar feature (see for example \cite{Kreimer:0509135}
for detailed example computations).)
\end{blanko}

\begin{blanko}{Remark in defence of $B_+$-operators that are not quite
  Hochschild $1$-cocycles.}\label{notHH}
  It should be noted that neither the individual
  $B_+$-operators, nor the global one, are Hochschild $1$-cocycles
  for the twisted Hochschild cohomology introduced by Connes and 
  Kreimer~\cite{Connes-Kreimer:9808042}.
  The twist in this notion of Hochschild cohomology is the fact
  that the bicomodule structure on a bialgebra $H$ is not the
  standard one given by comultiplication, but that the right coaction
  involves $\eta\epsilon$ instead of the identity
  (see Moerdijk~\cite{Moerdjik:9907010}).
  
  The corresponding $1$-cocycle condition for $B_+$ would be
  $$\xymatrix{
     \BB \ar[r]^-\Delta\ar[d]_{B_+} & \BB\tensor \BB 
     \ar[d]^{\Id \tensor B_+ + B_+ \tensor \eta\epsilon} \\
     \BB \ar[r]_-\Delta & \BB\tensor \BB,
  }$$
  which fails to hold because of one term: by applying first $B_+$ to a forest
  and then cutting up the resulting tree with $\Delta$ produces (among others) 
  one term corresponding to cut below the new node, providing a
  right-hand term which is a trivial tree.  On the other hand, by first
  comultiplying and then applying the right-hand vertical map, the second
  summand $B_+ \tensor \eta\epsilon$ produces instead the term $1$.
  The failure to be a Hochschild $1$-cocycle is thus intimately related
  to the failure of $\BB$ to be connected, cf.~\ref{bla:nonconnected}.
  In fact it is clear that the $B_+$-operators descend to the connected
  quotient of $\BB$, and that for this Hopf algebra they {\em do} become
  Hochschild $1$-cocycles.
  
  The Hochschild $1$-cocycle condition is emphasised in numerous papers by
  Kreimer 
  for ensuring locality of counter-terms and finiteness of renormalisation.  A
  rather detailed proof can be found in \cite{Kreimer:0306020}, with additional
  remarks in \cite{Kreimer:0202110}.  There are two arguments in the proof that
  use the $1$-cocycle condition: one is that every relevant tree is in the image
  of some $B_+$-operator, and the other is that the $1$-cocycle condition allows
  for a proof by induction.  Inspection of these two arguments shows that they
  work also for the not-quite-$1$-cocycles in $\BB_P$.  That every tree is in
  the image of some $B_+$ is tautological in the present approach.  For the
  induction argument, note again that the failure of the $1$-cocycle condition
  amounts to the appearance of a nodeless tree instead of a unit.  But since
  nodeless trees are still group-like, they can be used as basis for the
  induction instead of the algebra unit $1$, as detailed in \cite{Kock:1411.3098}.

  Another important consequence attributed to the $1$-cocycle condition is
  Theorem~\ref{thm:BK}, the fact that the homogeneous parts of the solution to
  the Dyson--Schwinger equation form a Hopf sub-algebra.  We shall see below
  (\ref{thm:GKT}) that this is ensured already by the inductive structure of the
  solution, seen to be more fundamental than the question of whether the $B_+$
  are Hochschild $1$-cocycles.
\end{blanko}
  
\begin{blanko}{Green functions.}
  While taking core immediately puts us into the familiar Connes--Kreimer Hopf
  algebra, it is actually important not to throw away the information encoded 
  in leaves and root: the
  strict type obedience characteristic for $P$-trees 
  (respect for arities) allows for meaningful automorphism
  groups and the existence of meaningful Green 
  functions~\cite{GalvezCarrillo-Kock-Tonks:1207.6404}.  In the abstract 
  setting of~\cite{GalvezCarrillo-Kock-Tonks:1207.6404} and 
  \cite{Kock-Weber:1609.03276}, the (connected) Green function is simply the
  sum of all connected objects (weighted by symmetry factors).  In the present
  case, the objects are forests, and the connected forests are trees.
  At the groupoid level, symmetry factors are 
  accounted for automatically, and the connected Green function is simply the functor
  $$
  \tree \to \forest_1 ,
  $$
  considered as an object in $\Grpd_{/\forest_1}$.
  Since $\tree$ is an infinite groupoid, at the vector-space level,
  the connected Green function lives in the completion
  $\Q[[T\in \pi_0 \tree]]$ and is a formal series
  $$
  G = \sum_{T\in \pi_0 \tree} \frac{T}{\norm{\Aut(T)}},
  $$
  the sum is over iso-classes of $P$-trees, and it comes about really as
  a homotopy sum.
  The adjective `connected' is
  appropriate from the viewpoint of trees, but is misleading when interpreted
  in quantum field theory: depending on how the
  dictionary is set up (as in Section~\ref{sec:trees-graphs}), `connected'
  for trees corresponds to `1PI (and connected)' for graphs.  The Green function
  just introduced is the only one used in this work, so to avoid confusion we 
  just call it the {\em Green function}.  From the constructions, the following
  is clear.
\end{blanko}
  
\begin{prop}
  $G$ solves the Dyson--Schwinger equation.
\end{prop}
 
Furthermore, the Green function naturally splits into summands
$$
G = \sum_{n\geq 0} g_n   ,
$$
where $g_n$ consists of all the trees with $n$ leaves.
We have

\begin{theorem}\label{thm:GKT} (\cite{GalvezCarrillo-Kock-Tonks:1207.6404})
  The Green function satisfies the Fa\`a di Bruno formula
  $$
  \Delta(G) = \sum_{n\geq 0} G^n \tensor g_n .
  $$
\end{theorem}

An important point here is that 
the exponent $n$ on the left-hand tensor factor 
counts $n$ trees, each with a root, precisely matching the 
subscript $n$ in the right-hand tensor factor, which is the number of leaves in the 
trees in $g_n$.  This kind of information cannot be seen in at the level of 
combinatorial trees.

The theorem is proved in \cite{GalvezCarrillo-Kock-Tonks:1207.6404} at the
objective level, meaning that it is expressed as an equivalence of groupoids.
The coupling constant can be introduced without affecting the result, since in
any case it encodes precisely the grading with respect to which the
comultiplication is homogeneous.  It should be mentioned that in the case
involving more than one colour ($I\neq 1$), there are also separate Green
functions $G_v$ for each $v\in I$ (i.e.~for each residue),
each satisfying a Fa\`a di Bruno formula
(cf.~\cite[Theorem~7.8]{GalvezCarrillo-Kock-Tonks:1207.6404}).

\begin{blanko}{Example: trees.}\label{ex:naked}
  Let $\B$ denote the groupoid of finite sets and bijections, and let
  $\B'$ denote the groupoid of finite pointed sets and basepoint-preserving 
  bijections.
  The polynomial functor represented by $1 \leftarrow \B' \to \B \to 1$
  is the exponential functor
  \begin{eqnarray*}
    \Grpd & \longrightarrow & \Grpd  \\
    X & \longmapsto & \sum_{n\in\N} \frac{X^n}{\mathfrak S_n} .
  \end{eqnarray*}
  The $P$-trees for this functor are the naked trees (i.e.~abstract trees 
  without a planar structure or any other structure or decorations).
\end{blanko}

\begin{blanko}{Example: stable trees.}\label{ex:stable}
  In a similar vein, we can consider $P$-trees for the polynomial
  functor $P(X)=\exp(X)-1-X$, represented by
  $$
  1\leftarrow {\bf Y}'\to {\bf Y}\to 1,
  $$
  where $\bf Y$ is the groupoid of finite sets of cardinality at least $2$.
  These are naked trees with no nullary and no unary nodes, called {\em reduced
  trees} by Ginzburg and Kapranov~\cite{Ginzburg-Kapranov}.  We prefer the term
  {\em stable tree} \cite{GalvezCarrillo-Kock-Tonks:1207.6404}. 
  For a given number of leaves there is only a finite number
  of isoclasses of stable trees.

  Here are pictures of the isoclasses of stable trees with up to five leaves,
  with the symmetry factor of each tree indicated.
    \begin{equation}\label{eq:stable-trees}\begin{texdraw}
      \setunitscale 0.6
      \tiny
      
    \move (0 0) 
    \bsegment
      \rlvec (0 20) 
    \esegment
    
    \move (90 0) 
    \bsegment
      \move (0 0)
      \lvec (0 10) \onedot \rlvec (-5 10)
      \move (0 10) \rlvec (5 10)
      \move (6 0) \htext{$\brok{1}{2}$}
    \esegment
    
    \move (180 -5)   
    \bsegment
      \move (0 0)
          \bsegment
	    \move (0 0)
	    \lvec (0 10) \onedot \rlvec (-10 20)
	    \move (0 10) \rlvec (10 20)
	    \move (-5 20) \onedot \rlvec (5 10)
	    \move (6 0) \htext{$\brok{1}{2}$}

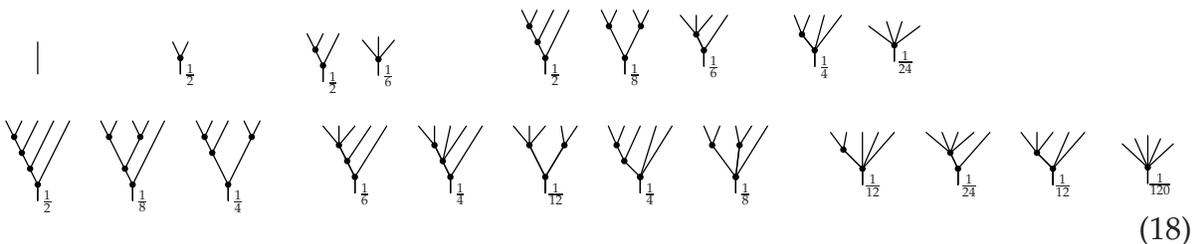
	  \esegment

	  \move (35 0)
	  \bsegment
	    \move (0 4)
	    \lvec (0 14) \onedot \rlvec (-10 12)
	    \move (0 14) \rlvec (10 12)
	    \move (0 14) \rlvec (0 14)
	    \move (6 4) \htext{$\brok{1}{6}$}

	  \esegment
    \esegment
    
    \move (320 0)
    \bsegment
      \move (0 0) 
      \bsegment
	\move (0 0) 
	\lvec (0 10) \onedot \rlvec (-15 30)
	\move (0 10) \rlvec (15 30)
	\move (-5 20) \onedot \rlvec (10 20)
	\move (-10 30) \onedot \rlvec (5 10)
	\move (6 0) \htext{$\brok{1}{2}$}

      \esegment
   
      \move (50 0) 
      \bsegment
	\move (0 0) 
	\lvec (0 10) \onedot \rlvec (-15 30)
	\move (0 10) \rlvec (15 30)
	\move (-10 30) \onedot \rlvec (5 10)
	\move (10 30) \onedot \rlvec (-5 10)
	\move (6 0) \htext{$\brok{1}{8}$}
      \esegment

      \move (100 0) 
      \bsegment
	\move (0 5) 
	\lvec (0 15) \onedot \rlvec (15 22)
	\move (0 15) \lvec (-5 25) 
	\onedot \rlvec (-10 12)
	\move (-5 25) \rlvec (10 12)
	\move (-5 25) \rlvec (0 12)
	\move (6 5) \htext{$\brok{1}{6}$}
      \esegment

      \move (170 0) 
      \bsegment
	\move (0 5) 
	\lvec (0 15) \onedot \lvec (-8 25) \onedot
	\rlvec (-5 12) \move (-8 25)
	\rlvec (5 12)
	\move (0 15) \lvec (6 37) 
	\move (0 15) \lvec (17 37) 
	\move (6 5) \htext{$\brok{1}{4}$}
      \esegment

      \move (220 0) 
      \bsegment
	\move (0 8) 
	\lvec (0 18) \onedot 
	\lvec (-16 30)
	\move ( 0 18) \lvec (-5 33)
	\move ( 0 18) \lvec (5 33)
	\move ( 0 18) \lvec (16 30)
	\move (7 8) \htext{$\brok{1}{24}$}
      \esegment

    \esegment


    \move (0 -80)
    \bsegment
      \move (0 0) 
      \bsegment
	\move (0 0) 
	\lvec (0 10) \onedot \rlvec (-20 40)
	\move (0 10) \rlvec (20 40)
	\move (-5 20) \onedot \rlvec (15 30)
	\move (-10 30) \onedot \rlvec (10 20)
	\move (-15 40) \onedot \rlvec (5 10)
	\move (6 0) \htext{$\brok{1}{2}$}
      \esegment
     
      \move (60 0)
      \bsegment
	\move (0 0) 
	\lvec (0 10) \onedot \rlvec (-20 40)
	\move (0 10) \rlvec (20 40)
	\move (-5 20) \onedot \rlvec (15 30)
	\move (5 40) \onedot \rlvec (-5 10)
	\move (-15 40) \onedot \rlvec (5 10)
	\move (6 0) \htext{$\brok{1}{8}$}
      \esegment
      
      \move (120 0)
       \bsegment
	\move (0 0) 
	\lvec (0 10) \onedot \rlvec (-20 40)
	\move (0 10) \rlvec (20 40)
	\move (-10 30) \onedot \rlvec (10 20)
	\move (15 40) \onedot \rlvec (-5 10)
	\move (-15 40) \onedot \rlvec (5 10)
	\move (6 0) \htext{$\brok{1}{4}$}
      \esegment
     
      \move (200 5) 
      \bsegment
	\move (0 0) 
	\lvec (0 10) \onedot \lvec (-10 30)
	\onedot 
	\rlvec (-10 12)
	\move (-10 30) \rlvec (10 12)
	\move (-10 30) \rlvec (0 12)
	\move (0 10) \lvec (20 42)
	\move (-5 20) \onedot \lvec (10 42)
	\move (6 0) \htext{$\brok{1}{6}$}
      \esegment

      \move (260 5) 
      \bsegment
	\move (0 0) 
	\lvec (0 10) \onedot \lvec (-10 30)
	\onedot 
	\rlvec (-10 12)
	\move (-5 20) \lvec (0 42)
	\move (-10 30) \rlvec (0 12)
	\move (0 10) \lvec (20 42)
	\move (-5 20) \onedot \lvec (10 42)
	\move (6 0) \htext{$\brok{1}{4}$}
      \esegment
      
      \move (320 5) 
      \bsegment
	\move (0 0) 
	\lvec (0 10) \onedot \lvec (-10 30)
	\onedot 
	\rlvec (-10 12)
	\move (-10 30) \rlvec (10 12)
	\move (-10 30) \rlvec (0 12)
	\move (0 10) \lvec (12 30) \onedot \lvec (10 42)
	\move (12 30) \lvec (20 42)
	\move (6.5 0) \htext{$\brok{1}{12}$}
      \esegment

      \move (380 5)
      \bsegment
	\move (0 0) 
	\lvec (0 10) \onedot \lvec (-10 20) \lvec (-15 30) \lvec (-20 42)
	\move (0 10) \lvec (20 42)
	\move (0 10) \onedot \lvec (10 42)
	\move (-10 20) \onedot \lvec (0 42)
	\move (-15 30) \onedot \lvec (-10 42)
	\move (6 0) \htext{$\brok{1}{4}$}
      \esegment

      \move (440 5)
      \bsegment
	\move (0 0) 
	\lvec (0 10) \onedot \lvec (-15 30) \lvec (-20 42)
	\move (-15 30) \onedot \lvec (-10 42)
	\move (0 10) \lvec (2.5 30) \onedot \lvec (0 42)
	\move (2.5 30) \lvec (10 42)
	\move (0 10) \onedot \lvec (20 42)
	\move (6 0) \htext{$\brok{1}{8}$}
      \esegment
      \move (520 10) 
      \bsegment
	\move (0 0) 
	\lvec (0 10) \onedot \lvec (-12 22) 
	\onedot \lvec (-20 33)
	\move (-12 22) \lvec (-10 33)
	\move (0 10) \lvec(0 33)
	\move (0 10) \lvec(10 33)
	\move (0 10) \lvec(20 33)
 	\move (6.5 0) \htext{$\brok{1}{12}$}
     \esegment
      \move (580 10) 
      \bsegment
	\move (0 0) 
	\lvec (0 10) \onedot \lvec (-5 20) 
	\onedot \lvec (-20 33)
	\move (-5 20) \lvec (-10 33)
	\move (-5 20) \lvec(0 33)
	\move (-5 20) \lvec(10 33)
	\move (0 10) \lvec(20 33)
	\move (7 0) \htext{$\brok{1}{24}$}
      \esegment
      \move (640 10) 
      \bsegment
	\move (0 0) 
	\lvec (0 10) \onedot \lvec (-10 20) 
	\onedot \lvec (-20 33)
	\move (-10 20) \lvec (-10 33)
	\move (-10 20) \lvec(0 33)
	\move (0 10) \lvec(10 33)
	\move (0 10) \lvec(20 33)
	\move (6 0) \htext{$\brok{1}{12}$}
      \esegment

      \move (700 3) 
      \bsegment
	\move (0 8) 
	\lvec (0 18) 
	\onedot \lvec (-16 33)
	\move ( 0 18) \lvec (-8 36)
	\move (0 18) \lvec (0 38)
	\move ( 0 18) \lvec (8 36)
	\move ( 0 18) \lvec (16 33)
	\move (7.7 8) \htext{$\brok{1}{120}$}
      \esegment

    \esegment

    \setunitscale 1
  \end{texdraw}\end{equation}

  Taking core yields the solution to the Dyson--Schwinger equation in 
  Example~\ref{ex:infinite-symmetric-DSE}.  Hence the coefficients there
  are explained as the number of stable trees with $k+1$ leaves, counted with
  symmetry factors.
\end{blanko}

More generally we have the following theorem, whose proof is straightforward.
\begin{theorem}
  Consider any finitary polynomial endofunctor $P$ (assumed to have no nullary and no unary 
  operations), and its corresponding abstract Dyson--Schwinger equation
  $$
  X \isopilback 1 + P(X),
  $$
  solved by $\tree_P$, the groupoid of $P$-trees.
  Write down the corresponding Bergbauer--Kreimer style equation in $\CK$,
  with solution $\sum_k c_k \alpha^k$.  Then the coefficient of a combinatorial
  tree $t$ in $c_k$ is the number of $P$-trees with $k+1$ leaves and core $t$.
\end{theorem}

\section{Natural transformations}
\label{sec:nat}

\begin{blanko}{Homotopy natural transformations.}\label{NT}
  Given two polynomial endofunctors 
  $$
  P,P': \Grpd_{/I} \to \Grpd_{/I},
  $$
  a {\em homotopy natural transformation} from $P'$ to $P$
  is the data of:
  for each object $X$ in $\Grpd_{/I}$ a morphism $u_X: P'(X) \to P(X)$ in 
  $\Grpd_{/I}$,
  and for each morphism $f:X \to Y$ in $\Grpd_{/I}$ a $2$-cell
  $$\xymatrix{
     P'(X) \ar[r]^{u_X}\ar[d]_{P'(f)} \ar@{}[rd]|\Rightarrow& P(X) \ar[d]^{P(f)} \\
     P'(Y) \ar[r]_{u_Y} & P(Y) ,
  }$$
  required to satisfy various compatibility conditions.
  It is called a {\em cartesian} natural transformation (the word `homotopy' being 
  dropped) when for every $f:X \to Y$ this square is a homotopy pullback.
  
  Just as in the classical case~\cite{Gambino-Kock:0906.4931}, cartesian
  natural transformations $P'\Rightarrow P$ correspond precisely to 
  diagrams
  $$    \xymatrix{
  I \ar[d]_=& \ar[l]  E'\ar[d] \drpullback \ar[r] & B'\ar[d] \ar[r] 
& I\ar[d]^= \\
  I  &\ar[l] E\ar[r] & B \ar[r]  &I ,
  }$$
  the middle square being required to be a homotopy pullback.
\end{blanko}

\begin{blanko}{Cartesian morphisms.}\label{cart}
  More generally there is a notion of cartesian morphism between
  polynomial endofunctors defined on different slices, which amounts to diagrams
  \begin{equation}\label{eq:cart}\xymatrix{
  I' \ar[d]& \ar[l]  E'\ar[d] \drpullback \ar[r] & B'\ar[d] \ar[r] 
& I'\ar[d] \\
  I  &\ar[l] E\ar[r] & B \ar[r]  &I .
  }\end{equation}
  With the interpretation of the groupoids $B'$ and $B$ as spaces of 
  operations, we see that the pullback condition on the middle square amounts
  to requiring that each operation of $P'$ is sent to an operation of $P$
  with the same arity.  The outer squares express compatibility with colours.
  
  In particular, as was already mentioned, the notion of $P$-tree is a special 
  case of cartesian morphism of polynomial endofunctors.  For the same reason,
  the notion of cartesian morphism of polynomial endofunctors is precisely the
  notion of morphism that behaves well with respect to decorating trees:
  if $P' \Rightarrow P$ is a cartesian morphism of polynomial endofunctors,
  then there is induced a morphism of groupoids from $P'$-trees to $P$-trees,
  given simply by composing diagrams like \eqref{eq:cart}.  This amounts
  to redecoration of the nodes and edges along the vertical maps in 
  \eqref{eq:cart}.
  
  Furthermore, this morphism is compatible
  with the algebraic structure we have set up:
\end{blanko}

\begin{prop}
  A cartesian morphism of polynomial functors $u:P' \Rightarrow P$
  induces a bialgebra homomorphism $\BB_{P'} \to \BB_P$, which is 
  compatible with $B_+$-operators in the sense that for each $b' \in B'$
  the following square commutes:
  $$\xymatrix{
     \forest_{P'} \ar[r]\ar[d]_{b'} & \forest_{P} \ar[d]^{u(b')} \\
     \tree_{P'} \ar[r] & \tree_P .
  }$$  
\end{prop}

\begin{proof}
  (Cf.~\cite{GalvezCarrillo-Kock-Tonks:1207.6404}.)
  It is simply a question of redecorating trees.
\end{proof}

There is also a compatibility with Green functions, but it is contravariant:
\begin{prop}
  For $u: P' \Rightarrow P$ a cartesian morphism of polynomial endofunctors
  we have
  $$
  G' = u\upperstar G,
  $$
  where $G'$ is the Green function in $\BB_{P'}$ and $G$ is the 
  Green function in $\BB_P$.
\end{prop}

\begin{proof}
  This follows from the definition of the Green functions as the fully faithful
  functors $\tree \to \forest$, and the obvious fact that this square is a
  homotopy pullback:
  $$\xymatrix{
     \tree_{P'} \drpullback \ar[r]\ar[d] & \tree_P \ar[d] \\
     \forest_{P'} \ar[r] & \forest_P .
  }$$
\end{proof}

\begin{blanko}{Examples.}
  Let $\C$ denote the groupoid of finite cyclically ordered sets and
  $\C'$ the groupoid of pointed cyclically ordered finite sets.
  For the $1$-variable polynomial functor represented by
  $$
  \C' \to \C ,
  $$
  the corresponding notion of tree is that of cyclic trees, i.e.~trees with a 
  cyclic ordering of the incoming edges of each node.
  We also saw the polynomial
  $$
  \N' \to \N
  $$
  whose trees are planar trees (\ref{ex:planar}), and the polynomial
  $$
  \B'\to\B
  $$
  whose trees are naked trees (\ref{ex:naked}).
  The diagram of groupoids
\[
\xymatrix{
\N' \ar[r]\ar[d]\drpullback & \N \ar[d] \\
\C' \ar[r]\ar[d]\drpullback & \C \ar[d] \\
\B' \ar[r] & \B 
}\]
now represents cartesian natural transformations,
and they induce the evident morphisms from planar trees to cyclic
trees to naked trees,
and the corresponding bialgebra homomorphisms.
\end{blanko}

\begin{blanko}{Subfunctors.}
  A useful special case of cartesian morphism is that of subfunctors.
  A subfunctor is simply a natural transformation such that for each object in the 
  domain, the associated map is injective.  These are always cartesian.
\end{blanko}

\begin{blanko}{Example.}
  The inclusion of $X^2$ into $X^0+X^1+X^2+X^3+\cdots$ is the inclusion of
  binary planar trees into all planar trees.
\end{blanko}

\begin{blanko}{Truncation.}
  From the viewpoint of Dyson--Schwinger equations, subfunctors implement a good
  notion of truncation, which is an important aspect of the theory of
  Dyson--Schwinger equations.  In realistic situations in quantum field theory,
  the polynomial will have infinitely many terms, for example by listing
  infinitely many primitive graphs for a given theory.  
  For practical reasons one may be forced to truncate this sequence,
  allowing only a finite number of primitive graphs.
  It is a prominent aspect of many current developments, notably in
  Quantum Chromodynamics, to select these truncations carefully so as to maintain
  certain qualitative properties of the solutions.
  Examples of much studied truncations are the rainbow and ladder-rainbow 
  approximations to Yukawa theory.  Many other examples can be found in the
  survey of Roberts~\cite{Roberts:1203.5341}.
\end{blanko}

\section{Foissy equations}

\label{sec:Foissy}

\begin{blanko}{Foissy equations.}
  Foissy~\cite{Foissy:0707.1204} has studied extensively a different form of
  combinatorial Dyson--Schwinger equations
  (which is required to take place inside the Hopf algebras of trees,
  or actually, the non-commutative Hopf algebra of planar trees and forests),
  \begin{equation}\label{eq:Foissy}
  Y = \alpha B_+ (f(Y)),
  \end{equation}
  referring to a formal series $f(h)=\sum_{k=0}^\infty p_k h^k$
  required to start with $p_0 = 1$.

  There are two notable differences compared to the Bergbauer--Kreimer 
  equations~\eqref{eq:BK-DSE}.
  The first is the placement of the constant $1$: in the Bergbauer--Kreimer
  equation, this constant is outside the $B_+$-operator, whereas in the Foissy
  equation~\eqref{eq:Foissy} it is present as the requirement that the series
  $f(h)$ starts with $1$.
  
  The second, more substantial difference is the role of the coupling constant
  $\alpha$: while in the Bergbauer--Kreimer equation the power of $\alpha$
  tracks the operadic grading (which really only becomes visible in the present
  formalism), in the Foissy equation the power of $\alpha$ tracks the node
  grading.  As we proceed to explain, both differences are illuminated by the
  groupoid setting.
\end{blanko}

\begin{blanko}{Foissy series.}\label{ex:Foissy}
  While a unique solution always exists for the Foissy Dyson--Schwinger
  equations~\eqref{eq:Foissy}, for the same reason as in the Bergbauer--Kreimer
  case, it is {\em not} always true that the solution spans a Hopf sub-algebra.
  Foissy~\cite{Foissy:0707.1204} figured out exactly for which $f$ the solution
  spans a Hopf sub-algebra.  Subsequent papers \cite{Foissy:0909.0358},
  \cite{Foissy:1112.2606} generalised this to the multivariate setting.  These
  {\em Foissy series} $f$ come in continuous families, which include the
  following three cases with non-negative coefficients.
  
  \begin{enumerate}
  \item The exponential series $f(h) = \exp(h)$.
   In this case, the equation thus reads
   $$
   Y = \alpha B_+( \exp(Y) ),
   $$
   and with the Ansatz $Y = \sum_{n\geq 1} a_n \alpha^n$, the solution is found 
   to be
   $$
  a_1 = \raisebox{1pt}{\ctreeone} , \qquad
  a_2 = \raisebox{-2pt}{\ctreetwo} , \qquad
  a_3 = \raisebox{-4pt}{\ctreethreeL} + \brok{1}{2} \raisebox{-2pt}{\ctreethreeV} , \qquad
  a_4 =  \raisebox{-8pt}{\ctreefourL} 
  + \brok{1}{2} \!\raisebox{-4pt}{\ctreefourY} +  \raisebox{-4pt}{\ctreefourVL} +
  \brok{1}{6} \raisebox{-2pt}{\ctreefourW} ,
  \phantom{xxxxxxxxx}
  $$
  $$\phantom{x}
  a_5 = 
   \raisebox{-12pt}{\ctreefiveL} + \brok{1}{2} \!\raisebox{-8pt}{\ctreefiveIY}
  +  \raisebox{-8pt}{\ctreefiveYL} + \raisebox{-8pt}{\ctreefiveVLL}
  + \brok{1}{2}\!\raisebox{-4pt}{\ctreefiveVlV}
  + \brok{1}{2}\raisebox{-4pt}{\ctreefiveVII}
  + \brok{1}{6}\!\!\raisebox{-4pt}{\ctreefiveIW}
  + \brok{1}{2}\raisebox{-4pt}{\ctreefiveWL}
  + \brok{1}{24}\raisebox{-4pt}{\ctreefivewide} ,
  \qquad \text{etc.}
  $$
  
  \item The geometric series
  $f(h) = \frac{1}{1-h}$.
  In this case, the equation thus reads
  $$
  Y = \alpha B_+( \brok{1}{1-Y} ),
  $$
  and with the Ansatz $Y = \sum_{n\geq 1} a_n \alpha^n$, the solution is found 
  to be
  $$
  a_1 = \raisebox{1pt}{\ctreeone} , \qquad
  a_2 = \raisebox{-2pt}{\ctreetwo} , \qquad
  a_3 = \raisebox{-4pt}{\ctreethreeL} + \raisebox{-2pt}{\ctreethreeV} , \qquad
  a_4 =  \raisebox{-8pt}{\ctreefourL} 
  + \raisebox{-4pt}{\ctreefourY} +  2 \raisebox{-4pt}{\ctreefourVL} 
  + \raisebox{-2pt}{\ctreefourW} , 
  \phantom{xxxxxxx}
  $$
  $$\phantom{xxx}
  a_5 = 
   \raisebox{-12pt}{\ctreefiveL} + \raisebox{-8pt}{\ctreefiveIY}
  + 2 \raisebox{-8pt}{\ctreefiveYL} + 2 \raisebox{-8pt}{\ctreefiveVLL}
  + 2 \!\raisebox{-4pt}{\ctreefiveVlV}
  + \raisebox{-4pt}{\ctreefiveVII}
  + \raisebox{-4pt}{\ctreefiveIW}
  + 3 \raisebox{-4pt}{\ctreefiveWL}
  + \raisebox{-4pt}{\ctreefivewide} ,
  \qquad \text{etc.}
  $$
  
  \item For each $n\in \N$, the polynomial $f(h) = (1+h)^n$.
  As an example, with $n=3$ the equation reads
   $$
   Y = \alpha B_+( (1+Y)^3 ),
   $$
   and with the Ansatz $Y = \sum_{n\geq 1} a_n \alpha^n$, the solution is found 
   to be
   $$
  a_1 = \raisebox{1pt}{\ctreeone} , \qquad
  a_2 = 3 \raisebox{-2pt}{\ctreetwo} , \qquad
  a_3 = 9 \raisebox{-4pt}{\ctreethreeL} + 3 \raisebox{-2pt}{\ctreethreeV} , \qquad
  a_4 =  27 \raisebox{-8pt}{\ctreefourL} 
  + 9 \!\raisebox{-4pt}{\ctreefourY} +  18 \raisebox{-4pt}{\ctreefourVL} 
  + \raisebox{-2pt}{\ctreefourW} ,
  \phantom{xxxxxx}
  $$
  $$ \phantom{xx}
  a_5 = 
   81 \raisebox{-12pt}{\ctreefiveL} + 27 \!\raisebox{-8pt}{\ctreefiveIY}
  + 54 \raisebox{-8pt}{\ctreefiveYL} + 54 \raisebox{-8pt}{\ctreefiveVLL}
  + 18 \! \raisebox{-4pt}{\ctreefiveVlV}
  + 27 \raisebox{-4pt}{\ctreefiveVII}
  + 3 \!\! \raisebox{-4pt}{\ctreefiveIW}
  + 9 \raisebox{-4pt}{\ctreefiveWL} ,
  \qquad \text{etc.}
  $$
  \end{enumerate}
\end{blanko}

\begin{blanko}{Foissy-style fixpoint equations in groupoids.}
  \label{bla:Foissy}
  Given a polynomial endofunctor $P$, represented by
  $$
  I \leftarrow E \to B \to I ,
  $$
  this paper concentrates on the fixpoint equation
  \begin{equation}\label{eq:1+P}
    X \isopilback 1 + P(X).
  \end{equation}
  (Recall that $1$ denotes the terminal object
  in $\Grpd_{/I}$, the trivial family $I \to I$.)
  A variant is the equation
  \begin{equation}\label{eq:PY}
  X \isopilback P(X)  .
  \end{equation}
  
  The first equation can
  be considered a special case of the second equation, by simply considering the
  polynomial functor $1+P$.  This amounts to interpreting $1$ as a nullary 
  operation, or more precisely: a collection of
  nullary operations, one for each colour $i \in I$.
  On the other hand, it should be noted that if $P$ has no nullary operations,
  then $\emptyset = P(\emptyset)$, hence clearly 
  $\emptyset$ is the least fixpoint in this case (and not an interesting 
  one).  More precisely, if for some colour $i \in I$ the polynomial $P$ has no nullary
  operation of that colour, the solution will be supported away from that 
  colour as well, and the whole equation could be formulated with less colours.
  All told, it is reasonable to assume that $P$ has at least one nullary 
  operation of each colour, and hence is essentially of the form $1+P$.
  So the two forms of fixpoint equation have essentially the same content
  (as stated more formally in Proposition~\ref{prop:Q=1+P} below).
  
  The first form is favoured because of Theorem~\ref{thm:hoinit-gr}, and because
  of the rich structure of $P$-trees.  Furthermore, in a precise sense,
  initial-algebra formation is just a special case of free-monad formation
  \cite{Kock:NotesOnPolynomialFunctors}.  The free monad on $P$, denoted
  $\overline P$, is again polynomial; it has the following description on
  objects: $\overline P(S)$ is the initial $(S+P)$-algebra.  For any polynomial
  functor $P$ there is a canonical equivalence $P(1)\simeq B$.  In particular,
  the space of operations of $\overline P$ is $\overline P(1)$, which is the
  initial $(1+P)$-algebra, not the initial $P$-algebra.
\end{blanko}

\begin{blanko}{Dead $P$-trees.}
  A {\em dead $P$-tree} is a $P$-tree without leaves.
\end{blanko}

\begin{prop}
  The polynomial fixpoint equation
  $$
  X \isopilback P(X)
  $$
  has a least solution, namely the groupoid of dead $P$-trees.
\end{prop}

\begin{proof}
  This can be checked by the same line of argument as the proof of
  Theorem~\ref{thm:hoinit-gr}.
\end{proof}
Note that for the existence of a dead $P$-tree, it is necessary that $P$ has at
least one nullary operation.  Indeed, when $P$ has no constant term, obviously
$\emptyset$ is a solution, and obviously the least such.

\begin{prop}\label{prop:Q=1+P}
  For $Q=1+P$, there is a canonical equivalence between the groupoid of 
  $P$-trees and the groupoid of dead $Q$-trees.  More precisely the
  equations 
  $$
  X \isopilback 1 + P(X) \qquad \text{ and } \qquad Y \isopilback Q(Y)
  $$
  have equivalent solutions.
\end{prop}

\begin{proof}  
  We interpret the  $1$-summand in $Q$ as a family of nullary operations,
  for example pictured as white nodes.
  Then the equivalence $X \simeq Y$ is given
  by sending a $P$-tree to the $Q$-tree obtained by grafting one of those
  new nullary white nodes onto each leaf.  In the other direction, given
  a $Q$-tree, prune all its white nullary nodes.  It is clear that this
  gives one-to-one correspondence, and no automorphisms are created or killed,
  since the new nullary nodes are clearly automorphism-free.
\end{proof}

Although the dead trees  resemble the 
combinatorial trees of the classical Connes--Kreimer Hopf algebra,
the comultiplication that can be transported to dead trees via
the equivalence of Proposition~\ref{prop:Q=1+P} is not the Connes--Kreimer
comultiplication.  Indeed, the cuts involved in describing this transported
comultiplication are not allowed to go above the nullary nodes, and every
cut, rather than deleting edges, replaces them with nullary nodes.
It seems in fact at the moment that this new comultiplication is not
particularly useful, and that it is better to stick with the canonical
comultiplication of $P$-trees.

The following seems more useful.  Let $\BB = \BB_P$ denote the bialgebra of
$P$-trees.

\begin{prop}
  Dead $P$-trees form a right $\BB$-comodule algebra $\DD$, more precisely a right 
  $\BB$-coideal algebra, and in fact a graded coideal algebra for either grading.
\end{prop}

\begin{proof}
   This amounts essentially to definitions.
   Recall (e.g.~from \cite{Jonah:1968}) that for a bialgebra $B$, a
   right $B$-comodule algebra is an algebra $M$ together with a coaction $M \to
   M \tensor B$ required to be an algebra homomorphism.  The algebra structure
   on $\DD$ is disjoint union (induced from $\BB$): clearly the disjoint union
   of dead forests is again a dead forests.  The coaction is just the
   comultiplication restricted from $\BB$ to $\DD$: when cutting a dead tree
   $T$, the crown forest $P_c(T)$ is again dead, whereas the bottom tree
   $R_c(T)$ is an arbitrary tree, not in general dead.  (There is no canonical
   way of making it dead: we cannot just remove the leaves of the tree, as that
   would destroy the $P$-decoration, neither is it possible to add nullary
   operations to the leaves, as that would mess up with coassociativity.)
\end{proof}

\begin{blanko}{Taking core.}
  Observe that upon taking core, the difference between $P$-trees and dead 
  $P$-trees goes away.  In particular, the right $\BB$-comodule algebra $\DD$
  is transformed into the Connes--Kreimer Hopf algebra just like $\BB$ itself.
  In this way, we can recover the three Foissy equations from \ref{ex:Foissy}.
\end{blanko}

\begin{blanko}{Example: naked trees.}
  Consider the Foissy equation
  $$
  Y \isopilback \exp(Y) .
  $$
  By Proposition~\ref{prop:Q=1+P}, the solution $Y$ is the groupoid of dead 
  trees.  Taking core does nothing (except removing the 
  root edge), and upon sorting by node degree, 
  the solution is precisely that of Foissy (\ref{ex:Foissy} (1)), with the coefficient
  of a given tree being the inverse of its symmetry factor.
\end{blanko}

\begin{blanko}{Example: planar trees.}
  For the Foissy equation
  $$
  Y \isopilback \frac{1}{1-Y}
  $$
  the solution $Y$ is the groupoid of dead planar trees (which, since planar
  trees have no automorphisms, is just the set of isoclasses of planar trees).
  Taking core does nothing (except removing the root edge), and upon sorting by
  node degree, the solution is precisely that of Foissy (\ref{ex:Foissy} (2)), with the coefficient of
  a given tree being the number of different planar structures.
\end{blanko}

\begin{blanko}{Example: ternary planar trees.}
  For the Foissy equation
  $$
  Y \isopilback (1+Y)^3
  $$
  the solution $Y$ is the groupoid of dead $(1+Y)^3$-trees, which we analyse a 
  little further.  The polynomial functor $Y \mapsto (1+Y)^3$ is the composite
  of $Y \mapsto 1+Y$ and $X \mapsto X^3$. The first has a nullary operation
  \begin{texdraw}
    \move (0 0) \lvec (0 6) \move (0 7.5)\lcir r:1.5
  \end{texdraw}
  and a unary 
  operation
\inlineDotlessTree  ; the second has only a ternary operation
\begin{texdraw}
  \move (0 0) \lvec (0 6) \fcir f:0 r:1.5 \lvec (-4 10) 
  \move (0 6) \lvec (0 12) 
    \move (0 6) \lvec (4 10) 
\end{texdraw}
.
  By general theory of polynomial 
  functors~\cite{Kock:NotesOnPolynomialFunctors}, the operations of the composite are therefore two-level trees
  with bottom node the ternary operation of $X^3$, and second-level nodes
  the operations of $1+Y$:
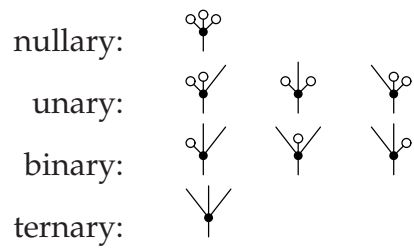
  \begin{center}
      \begin{tabular}{ rc l }
  nullary: && 
\begin{texdraw}
  \setunitscale{1.2}
  \bsegment
  \move (0 0) \lvec (0 6) \fcir f:0 r:1.5 
  \lvec (-3 9) \move (-4 10) \lcir r:1.5
  \move (0 6) \lvec (0 10) \move (0 11.5) \lcir r:1.5
    \move (0 6) \lvec (3 9) \move (4 10) \lcir r:1.5
    \esegment
\end{texdraw}
\\
  unary: &&
\begin{texdraw}
  \setunitscale{1.2}
  \move (0 0)
\bsegment
  \move (0 0) \lvec (0 6) \fcir f:0 r:1.5 
  \lvec (-3 9) \move (-4 10) \lcir r:1.5
  \move (0 6) \lvec (0 10) \move (0 11.5) \lcir r:1.5
    \move (0 6) \lvec (7 15) 
    \esegment
    \move (30 0)
\bsegment
    \move (0 0) \lvec (0 6) \fcir f:0 r:1.5 
  \lvec (-3 9) \move (-4 10) \lcir r:1.5
  \move (0 6) \lvec (0 16) 
    \move (0 6) \lvec (3 9) \move (4 10) \lcir r:1.5
    \esegment
    \move (60 0)
\bsegment
    \move (0 0) \lvec (0 6) \fcir f:0 r:1.5 
  \lvec (-7 15) 
  \move (0 6) \lvec (0 10) \move (0 11.5) \lcir r:1.5
    \move (0 6) \lvec (3 9) \move (4 10) \lcir r:1.5
    \esegment
\end{texdraw} \\
  binary: &&
\begin{texdraw}
  \setunitscale{1.2}
  \move (0 0)
\bsegment
  \move (0 0) \lvec (0 6) \fcir f:0 r:1.5 
  \lvec (-3 9) \move (-4 10) \lcir r:1.5
  \move (0 6) \lvec (0 16) 
    \move (0 6) \lvec (7 15) 
    \esegment
    \move (30 0)
\bsegment
    \move (0 0) \lvec (0 6) \fcir f:0 r:1.5 
  \lvec (-7 15) 
  \move (0 6) \lvec (0 10) \move (0 11.5) \lcir r:1.5
    \move (0 6) \lvec (7 15) 
    \esegment
    \move (60 0)
\bsegment
    \move (0 0) \lvec (0 6) \fcir f:0 r:1.5 
  \lvec (3 9) \move (4 10) \lcir r:1.5
  \move (0 6) \lvec (0 16) 
    \move (0 6) \lvec (-7 15) 
    \esegment
\end{texdraw} \\
   ternary: &&
\begin{texdraw}
  \setunitscale{1.2}
  \move (0 0)
\bsegment
    \move (0 0) \lvec (0 6) \fcir f:0 r:1.5 
  \lvec (-7 15)
  \move (0 6) \lvec (0 16) 
    \move (0 6) \lvec (7 15) 
    \esegment
\end{texdraw} \\
  \end{tabular}
  \end{center}
(corresponding to the $8$ terms in $(1+Y)^3 = 1 + 3Y+3Y^2+Y^3$).  
  The $(1+Y)^3$-trees are therefore all trees that can be built out of these 
  $8$ building blocks (remembering that the number of nodes is the number of such 
  operations, i.e.~the number of ternary nodes), 
  and the dead ones are those terminated everywhere with \begin{texdraw}
  \move (0 0) \lvec (0 6) \fcir f:0 r:1.5 
  \lvec (-3 9) \move (-4 10) \lcir r:1.5
  \move (0 6) \lvec (0 10) \move (0 11.5) \lcir r:1.5
    \move (0 6) \lvec (3 9) \move (4 10) \lcir r:1.5
\end{texdraw} .
  Under the correspondence of Proposition~\ref{prop:Q=1+P}, these correspond
  precisely to planar ternary trees, except that the trivial tree is left out,
  since we don't have \ \begin{texdraw} \move (0 0) \lvec (0 6) \move (0 7.5)
  \lcir r:1.5 \end{texdraw} .  Again, upon sorting by node degree, we see that
  the coefficients appearing in
  \ref{ex:Foissy}(3) count the number of planar ternary trees with a given core,
  in analogy with the general interpretation of the coefficients appearing in
  the Bergbauer--Kreimer equations (\ref{BK-DSE}--\ref{ex:infinite-symmetric-DSE}).
\end{blanko}

\section{Trees versus graphs}
\label{sec:trees-graphs}

This section is mostly a summary of the forthcoming paper 
\cite{Kock:graphs-and-trees}, which analyses the relationship
between Feynman graphs and $P$-trees.
It is not a logical part of the present paper, but it is important to
motivate the use of groupoids and the role of $P$-trees.

\bigskip

(All the pictures in this section are implicitly from massless $\phi^3$ theory
in six space-time dimensions.)

\begin{blanko}{Trees in BPHZ renormalisation.}\label{BPHZ}
  The main use of trees in BPHZ renormalisation is to express nestings of 
  Feynman graphs.  The discovery of Kreimer~\cite{Kreimer:9707029}
  was that the combinatorics of the BPHZ procedure is elegantly
  encoded in the Hopf algebra of rooted trees.
  More information, related to the 
  specifics of a particular theory, is encoded in the Hopf algebra of 
  graphs~\cite{Connes-Kreimer:9912092}.

  In order to understand the relationship between the two Hopf algebras
  well enough to transfer constructions and results such as the ones of the 
  present paper, some modifications seem necessary both on the graphs and the 
  tree side.  On the tree
  side, we pass to operadic trees, as already explained.
  In the following figure,
  
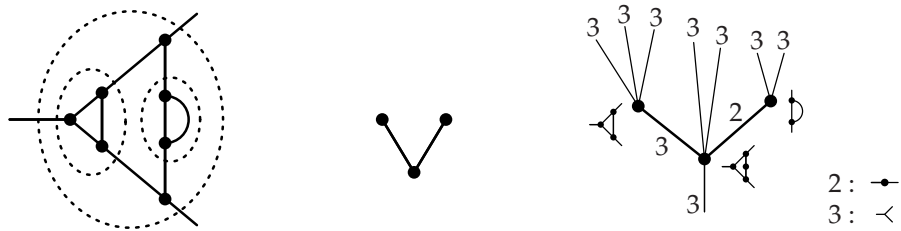
\begin{center}
  \begin{texdraw}
  \bsegment
	\linewd 1
      \move (-3 0) \lvec (20 0) \Onedot \lvec (68 40)
      \move (56 30) \Onedot \lvec (56 -30) \Onedot
      \move (32 10) \Onedot \lvec (32 -10) \Onedot
      \move (20 0) \lvec (68 -40)

      \move (56 9) \Onedot
      \move (56 -9) \Onedot
      \move (56 0) \larc r:9 sd:-90 ed:90

      \lpatt (1 3)
      \move (28 -1) \freeEllipsis{13}{20}{0}
      \move (58 0) \freeEllipsis{11}{16}{0}
      \move (43 0) \freeEllipsis{35}{41}{0}

  \esegment

  \move(150 -20)
  \bsegment
  	\linewd 1
  \move (0 0) \Onedot \lvec (-12 20) \Onedot
  \move (0 0) \lvec (12 20) \Onedot
  \esegment
  
\move (260 -35)

  \bsegment
  
    \move (0 0) 
    \lvec (0 20) \Onedot \lvec (-4 64) 
    \move (0 20) \lvec (6 62)
    \move (0 20) \linewd 1 
    \lvec (25 42) \linewd 0.5
    \Onedot \lvec (20 60) \move (25 42) \lvec (30 60)

    \move (0 20) \linewd 1 \lvec (-25 40) \linewd 0.5
 
    \bsegment 
      \move (0 0) \Onedot \lvec (-16 25) \move (0 0) \lvec (-5 30)
      \move (0 0) \lvec (6 29)
    \esegment

\move (6 17) \trekant 
\setunitscale{0.8} \rmove (12 0) \smalldot \setunitscale{1}
\move (-44 33) \trekant
\move (33 38) \tokant

\htext (-4 3){{\footnotesize $3$}}
\htext (-16 25){{\footnotesize $3$}}
\htext (-4 70){{\footnotesize $3$}}
\htext (7 67){{\footnotesize $3$}}

\htext (-42 70){{\footnotesize $3$}}
\htext (-30 75){{\footnotesize $3$}}
\htext (-18 74){{\footnotesize $3$}}

\htext (20 65){{\footnotesize $3$}}
\htext (30 65){{\footnotesize $3$}}
\htext (12 38){{\footnotesize $2$}}

\move (52 11)
\bsegment
\htext (0 0){{\footnotesize $2$ :}} \move (10 0) \tovert
\htext (0 -12){{\footnotesize $3$ :}} \move (10 -12) \trevert
\esegment
\esegment

\end{texdraw}
\end{center}
the small combinatorial tree in the middle expresses the nesting of
1PI subgraphs on the left; Kreimer showed that the information encoded by
such trees is sufficient to account for the counter-term corrections
of BPHZ. 

On the other hand, it is clear that such combinatorial trees
do not capture anything related to symmetries of graphs.
For this, fancier trees are needed,
as partially indicated on the right.
First of all, each node in the tree should be decorated by the 1PI graph
it corresponds to in the nesting~\cite{Bergbauer-Kreimer:0506190}, and second,
to allow an operadic interpretation,
the tree should have leaves (input
slots) corresponding to the vertices of the graph. 
Just as vertices of graphs serve as insertion points, the leaves of
a tree serve as input slots for grafting.
The decorated tree should be 
regarded as a
recipe for reconstructing the graph by inserting the decorating graphs into
the vertices of the graphs of parent nodes.
The numbers on the edges
indicate the  type constraint of each substitution: the outer interface of
a graph must match the local interface of the vertex it is substituted into.
But the type constraints on the tree decoration are not enough to reconstruct the
graph, because for example the small graph 
\raisebox{-5pt}{\begin{texdraw}\trekant\end{texdraw}} decorating
the left-hand node could be substituted into various different vertices of the 
graph
\raisebox{-5pt}{\begin{texdraw}
  \trekant  \setunitscale{0.8} \rmove (12 0) \smalldot \setunitscale{1}
\end{texdraw}}.
  The solution found in \cite{Kock:graphs-and-trees} is to consider $P$-trees,
  for $P$ a certain polynomial endofunctor over groupoids, which depends on the
  theory.  For this to work, a few modifications are needed on the graphs side:
\end{blanko}

\begin{blanko}{Adjustments to the Hopf algebra of graphs.}\label{Manchon}
  The first modification required is rather 
  harmless.  Traditionally, the Connes--Kreimer Hopf algebra of graphs
  is spanned by the 1PI graphs that are furthermore superficially divergent.  
  This last condition excludes the one-vertex graphs given by the interaction 
  labels themselves, and for this reason in the formula for comultiplication
  $$
  \Delta(\Gamma) = 1\otimes \Gamma + \Gamma \otimes 1 + \sum_{\emptyset\neq\gamma\subsetneqq\Gamma}
  \gamma \otimes \Gamma/\gamma
  $$
  the primitive part has to be specified separately since taking $\gamma=\Gamma$
  would yield a term with the one-vertex graph $\Gamma/\gamma = \res(\Gamma)$.
  While of course excluding the one-vertex graphs is natural from the viewpoint 
  of physics, from the strictly
  combinatorial viewpoint it appears as an ad hoc feature.  This was perhaps
  first observed by Manchon~\cite{Manchon:0408} who introduced a bigger
  bialgebra, by including the interaction labels (one-vertex graphs) as
  generators (see \cite{Kock:1411.3098} for further discussion).  
  Since obviously these new generators have loop number zero, this
  bigger bialgebra is not connected, and therefore no longer Hopf.  The
  difference is strictly analogous to the difference between reduced and
  non-reduced incidence algebras in Combinatorics (see~\cite{Galvez-Kock-Tonks:1612.09225});
  again, the standard Hopf algebra can be obtained by
  collapsing the degree-$0$ piece.  It should also be noted here that
  in this setting the sum constituting the Green function for a vertex $v$ does
  not start with $1$, but rather with $v$ considered as a graph with residue
  $v$.
  
  The second modification is subtler.  The Hopf algebra of graphs
  expresses contraction of subgraphs, but its dual Lie algebra is the one of insertions of
  graphs.  One is allowed to substitute graphs with two external legs into
  internal lines of the receiving graph.  This means that every internal line
  represents an ordered infinity of virtual insertion points.
  This does not look very good from the viewpoint of operadic trees, as it
  destroys the input slot correspondence between vertices in graphs and leaves
  in trees: the strict operadic viewpoint requires that grafting of trees only
  occurs at pre-existing leaves, and correspondingly in the setting of graphs,
  only insertions at vertices should be allowed.
  
  This can be arranged by declaring for each internal line a number of insertion
  points, for example by decorating it with special $2$-valent vertices
  (in the literature sometimes indicated with a cross).
  This means that among the primitive
  graphs we now have to include more different graphs, such as 
  
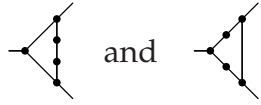
\begin{center}
  \begin{texdraw}
  \bsegment 
\move (-6 0) \lvec (0 0) \smalldot \lvec (18 18)
\move (0 0) \lvec (18 -18)
\move (12 12) \smalldot  \lvec (12 -12) \smalldot
\move (12 -4) \smalldot \move (12 4) \smalldot
\esegment
\htext(40 0){and}
\move (70 0)
  \bsegment 
\move (-6 0) \lvec (0 0) \smalldot \lvec (18 18)
\move (0 0) \lvec (18 -18)
\move (12 12) \smalldot  \lvec (12 -12) \smalldot
\move (6 -6) \smalldot \move (6 6) \smalldot
\esegment
\end{texdraw}
  \end{center}
and it also means that there is now a new type of vertex
(for each kinetic term in the Lagrangian), which is declared to be the
residue of any graph with two external lines, instead of saying 
that the residue is just the
  line.
  The Green function for the new interaction label $e= 
  \raisebox{2pt}{\begin{texdraw} \tovert 
\end{texdraw}}$ looks like this:
  $$
  G_e = \ \ \raisebox{0pt}{\begin{texdraw} 
\setunitscale{0.8}
\tovert 
\setunitscale{0.8}
\move(21 0) \htext{$+$} 
\move(32 0)\htext{$\frac12$} 
\move (38 0) \rundtokant 
\move(65 0)\htext{$+$} 
\move (72 0) \bsegment \rundtokant \move( 10.5 4) \smalldot \esegment
\move(100 0)\htext{$+$} 
\move(110 0)\htext{$\frac12$} 
\move (118 0) \bsegment \rundtokant \move( 10.5 4) \smalldot\move( 10.5 -4) \smalldot \esegment 
\move(150 0)\htext{$+$}
\move(165 0)\htext{$\ldots$} 
\end{texdraw}}
$$
and similarly, the Green functions for the proper interaction labels
are refined by all possible appearances of the new vertex.
These issues seem to be related to the $Z$-factor 
comparing the bare and normalised Green functions (see for example 
\cite{Itzykson-Zuber:QFT}, \cite{vanSuijlekom:0807}, 
\cite{EbrahimiFard-Patras:1003.1679}).
Let $\xxgraph$ denote the groupoid of connected 1PI graphs.
\end{blanko}

\begin{blanko}{Trees decorated with graphs.}
  Having been specific about which graphs we consider, we can now explain
  how to encode the graph-decorated trees, in order to get the correct
  correspondence between graphs and trees, cf.~\cite{Kock:graphs-and-trees}.
  The decorations are encoded as $P$-trees,
  for $P$ a certain polynomial endofunctor over groupoids, which depends on the
  theory. This formalism yields the correct symmetry factors.
  
  To match the figures above, we consider a theory in which there are
  two interaction labels \raisebox{2pt}{\begin{texdraw} \tovert \end{texdraw}} and
  \raisebox{-1pt}{\begin{texdraw} \trevert \end{texdraw}}\ ; let $I$ denote
  the groupoid of all such one-vertex graphs.  Let $B$ denote the groupoid
  of all connected 1PI graphs of the theory such that the residue belongs
  to $I$.  Finally let $E$ denote the groupoid of such graphs with a marked
  vertex.  The polynomial endofunctor $P$ is now given by the diagram
\begin{equation}
\xymatrix{
    I & \ar[l]_s  E  \ar[r]^p & B  \ar[r]^t & I ,
}
\end{equation}
  where the map $s$ returns the one-vertex subgraph at
  the mark, $p$ forgets the mark, and $t$ returns the residue of the
  graph, i.e.~the graph obtained by contracting everything to a point, but
  keeping the external lines.  A $P$-tree is hence a diagram 
    \begin{equation}\label{qtree}
  \xymatrix @!C=16pt {
  A\ar[d] \ar@{}[dr]|{\Leftarrow}& M\ar[l]  
  \ar@{}[dr]|{\Rightarrow}\drpullback\ar[r] \ar[d]& N \ar[d] \ar[r] \ar@{}[dr]|{\Rightarrow}
  &A\ar[d] \\
  I & \ar[l] E  \ar[r] & B \ar[r]&I \,, \\
}\end{equation}
  with specified invertible $2$-cells, in which the first row is a tree
  in the sense of~\ref{polytree-def}.
  These $2$-cells carry much of the structure: for
  example the $2$-cell on the right says that the 1PI graph decorating a given
  node must have the same residue as the decoration of the outgoing edge of the
  node --- or more precisely, and more realistically: an isomorphism is specified (it's
  a bijection between external lines of one-vertex graphs).  Similarly, the
  left-hand $2$-cell specifies for each node-with-a-marked-incoming-edge $x'\in
  M$, an isomorphism between the one-vertex graph decorating that edge and the
  marked vertex of the graph decorating the marked node $x'$.  Hence the
  structure of a $P$-tree is a complete recipe not only for which graphs should
  be substituted into which vertices, but also {\em how}: specific bijections
  prescribe which external lines should be identified with which lines in the receiving
  graph.  Let $\tree$ denote the groupoid of $P$-trees as in \eqref{qtree}.
\end{blanko}

\begin{blanko}{Graph nesting.}
  A {\em graph nesting} is a Feynman graph (assumed to be connected, 1PI and with residue
  belonging to $I$) with nested circles, such that every circle cuts a 1PI
  Feynman graph of the theory with residue in $I$.  The graph nestings form a
  groupoid $\nest$, in which the maps are graph isomorphisms compatible with the
  configuration of circles.
  
  The following is the main theorem of \cite{Kock:graphs-and-trees},
  which draws from insights from higher category 
  theory~\cite{Kock-Joyal-Batanin-Mascari:0706}.
  \begin{theorem}\label{N=T} {\em (\cite{Kock:graphs-and-trees})}
    There is an equivalence of groupoids between the groupoid $\nest$ of graph
    nestings and the groupoid $\tree$ of $P$-trees.  In particular, the symmetries
    of a given graph nesting can be read off the corresponding decorated tree
    and vice versa.
  \end{theorem}  
  It should be stressed that the use of groupoids as coefficients
  is crucial for getting the decorations that make this 
  correspondence work.  In fact, a tree decorated in groupoids may have
  more symmetries than the underlying tree.  For example, the graph $\Gamma = \ 
  $
  \raisebox{0pt}{\begin{texdraw}\rundtokant\end{texdraw}}
  is 1PI, and as a trivial nesting it corresponds to the tree
  \raisebox{-2pt}{\begin{texdraw}
\bsegment \setunitscale{0.8}
\move (0 -5) \lvec (0 2) \smalldot \lvec (-4 9) \move (0 2)  \lvec (4 9)
\esegment
\end{texdraw}} decorated with $\Gamma$ at the node, $3$ at the leaves, and $2$
at the root.  More formally it is of course a diagram like \eqref{qtree}.
It is straightforward to check that this $P$-tree has a symmetry
group of order $4$, just as the graph $\Gamma$, whereas the underlying tree
clearly has a symmetry group of order $2$.  
\end{blanko}

\begin{blanko}{Graphs with fixed residue.}
  So far we are talking abstract Feynman graphs (forming the groupoid
  $\xxgraph$), whereas in quantum field theory, the symmetries are required to
  fix the external lines.  Categorically, this means that we are talking about
  the groupoid $\xxgraph_v$ defined as the (homotopy) fibre over some residue
  $v$, in the running example $v = $ \raisebox{2pt}{\begin{texdraw} \tovert
  \end{texdraw}}~.  Inside this fibre, the symmetry group of the graph $\Gamma=$
  \raisebox{0pt}{\begin{texdraw}\rundtokant\end{texdraw}} is of order $2$, the
  non-trivial symmetry being the one that fixes the external lines and
  interchanges the two internal lines.  Since $v$ is an object in the groupoid
  $I$, the tree corresponding to $\Gamma$ belongs to the fibre $\tree_v$ of
  trees with root colour $v$.  One can check directly that in this groupoid, the
  tree has only one non-trivial automorphism, which in fact is trivial on the
  underlying tree!  Indeed, if we were to interchange the two leaves of the
  tree, then by the compatibilities expressed by the decoration, we would be
  interchanging the two vertices of $\Gamma$, and this in turn would interchange
  the two external lines of $v = $ \raisebox{2pt}{\begin{texdraw} \tovert
  \end{texdraw}}, the residue of $\Gamma$, but since we are inside the fibre
  $\tree_v$ this automorphism is not allowed.
\end{blanko}

\begin{blanko}{Nestings versus graphs.}
  By Theorem~\ref{N=T} we have an equivalence of groupoids $\nest \simeq \tree$.
  There is an obvious
  projection functor $\nest\to \xxgraph$ which simply forgets the circles
  expressing the nesting on a graph.  This functor is a finite discrete
  fibration --- this is just to say that for a given graph there is a finite set
  of possible nestings to put on it.  Pullback along this projection defines
  a functor
  $$
  \Grpd_{/\xxgraph} \longrightarrow \Grpd_{/\nest} \ \isopil\ \Grpd_{/\tree},
  $$
  which
  associates
  to each graph the set of possible nestings on it, and then the associated 
  $P$-tree.  Note that no coefficients appear in this sum, but as illustrated
  in \ref{breaking} below, there may be repetitions.
  This functor induces
  an algebra homomorphism
  $$
  \Psi: \Q[[\Gamma\in\pi_0\xxgraph]]\longrightarrow
\Q[[T\in\pi_0\tree]] .
$$

\begin{theorem}{\em (\cite{Kock:graphs-and-trees})}
    The map $\Psi:\Q[[\Gamma\in\pi_0\xxgraph]]\longrightarrow
\Q[[T\in\pi_0\tree]]$ is a bialgebra homomorphism.  Here
  the bialgebra structure on $\Q[[\Gamma\in\pi_0\xxgraph]]$ is the one of 
  \ref{Manchon};
    the bialgebra structure on $\Q[[T\in\pi_0\tree]]$ is the one of $P$-trees.
  \end{theorem}
\end{blanko}

\begin{prop}
  The bialgebra homomorphism $\Psi:\Q[[\Gamma\in\pi_0\xxgraph]]\longrightarrow
\Q[[T\in\pi_0\tree]]$ sends Green functions to Green functions.
\end{prop}

This is basically because the inverse image of the groupoid $\xxgraph_v$
is the whole groupoid $\nest_v$.  Since $\Psi$ is injective we conclude:
\begin{cor}
  The Fa\`a di Bruno formula holds for the Green function in the bialgebra of graphs.
\end{cor}
(This result has recently been established directly in the bialgebra of
graphs~\cite{Kock-Weber:1609.03276}, without reference to trees, as part of a 
more
general theory of Fa\`a di Bruno formalism.)

The following two examples illustrate the significance of this result.

\begin{blanko}{Example: nestings breaking symmetry.}\label{breaking}
First we consider the graph $\Gamma$ (with residue $v = $ \raisebox{2pt}{\begin{texdraw} \tovert 
\end{texdraw}}):

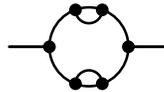
  \begin{center}
  \begin{texdraw}
    \move (0 0) \hhgraph
\end{texdraw}
\end{center}
The fibre over $\Gamma$ has four elements, named $N_1,N_2,N_3,N_4$:

  
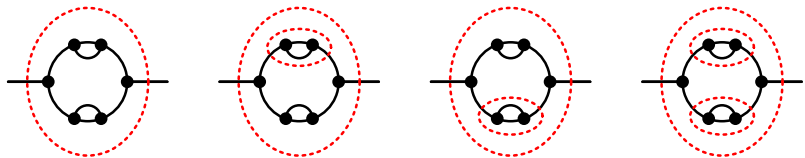
\begin{center}
  \begin{texdraw}
    \move (0 0) \bsegment \hhgraph \bigredoval \esegment
    \move (80 0) \bsegment \hhgraph \bigredoval \move (0 13) \redoval \esegment 
    \move (160 0) \bsegment \hhgraph \bigredoval \move (0 -13) \redoval\esegment
    \move (240 0) \bsegment \hhgraph \bigredoval \move (0 13) 
    \redoval\move (0 -13) \redoval\esegment
    
\end{texdraw}
\end{center}
The graph $\Gamma$ has an automorphism group of order $8$ (in $\xxgraph_v$), and thus appears 
in the Green function with a factor $\frac18$.
Hence $\Psi(\frac18\Gamma) = \frac18 N_1 + \frac18 N_2 + \frac18 N_3 + \frac18 
N_4$.  Now $N_2$ and $N_3$ are isomorphic in $\nest_v$, so we can also write the
sum as $\frac18 N_1 + \frac14 N_2 + \frac18 
N_4$, and these factors, $\frac18, \frac14 , \frac18$ are precisely the inverses
of the orders of the symmetry groups of the three objects in $\nest_v$.
What the example shows is the fact that symmetries of a graph can be broken by 
imposing nestings, but the decrease in symmetry is precisely counter-balanced
by the fact that a certain number of isomorphic nestings appear in the fibre.
\end{blanko}

\begin{blanko}{Example: overlapping divergences.}
  The second example concerns a graph with overlapping divergences.
Consider the graph $\Omega$ (with residue $v = $ \raisebox{2pt}{\begin{texdraw} \tovert 
\end{texdraw}}):
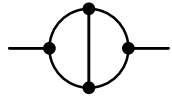
  \begin{center}
  \begin{texdraw}
    \move (0 0) \vgraph
\end{texdraw}
\end{center}
The fibre over $\Omega$ has three elements, denoted $N_1,N_2,N_3$:

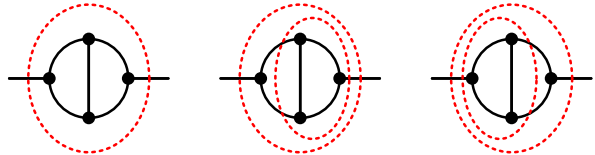
\begin{center}
  \begin{texdraw}
    \move (0 0) \bsegment \vgraph \bigredoval \esegment
    \move (80 0) \bsegment \vgraph \bigredoval \move (4.5 0) \redvoval \esegment 
    \move (160 0) \bsegment \vgraph \bigredoval \move (-4.5 0) \redvoval\esegment
  \end{texdraw}
\end{center}
In this case, $\Omega$ as well as the nestings $N_1,N_2,N_3$ all have an 
automorphism group of order $2$ (over $v$).  The interesting remark in this 
case is that the trees corresponding to the nestings $N_2$ and $N_3$ are {\em not}
isomorphic inside the fibre $\tree_v$ (although they {\em are} isomorphic as abstract
$P$-trees).  The reason for this is the observation already made earlier that
the drawings of these trees, even with all the decorating graphs, is not
the full picture.  Interchanging the two branches is only possible over
the non-trivial automorphism of $v = $ \raisebox{2pt}{\begin{texdraw} \tovert 
\end{texdraw}}.  (In fact it is clear in the drawings of nestings that
the two nestings are not isomorphic for fixed residue.)
\end{blanko}

\small


\hyphenation{mathe-matisk}

\label{lastpage}
\end{document}